\documentclass[12pt]{article}
\usepackage{amssymb,amsmath,latexsym,amsthm,amsfonts}
\usepackage[english]{babel}
\usepackage[T1]{fontenc}
\usepackage{color}
\usepackage{graphics}
\usepackage{graphicx}

\usepackage[colorinlistoftodos]{todonotes}

\newtheorem{theorem}{Theorem}[section]

\newtheorem{lemma}[theorem]{Lemma}
\newtheorem{proposition}[theorem]{Proposition}

\def\1g{1\hskip -3pt \mbox{l}}

\small\normalsize

    \title{Orthogonal Impulse Response Analysis in Presence of Time-Varying Covariance}

\author{
{\sc Valentin Patilea\footnote{CREST Ensai,
Campus de Ker-Lann,
51 Rue Blaise Pascal,
BP 37203, 35172 BRUZ Cedex
FRANCE; email: patilea@ensai.fr. This author acknowledges
support from the research program \emph{Mod\`eles et traitements math\'ematiques des donn\' ees en tr\`es grande dimension,} of Fondation du Risque and MEDIAMETRIE}
\qquad \qquad
Hamdi Ra\"{\i}ssi\footnote{Instituto de Estad\'{\i}stica, PUCV,
Errazuriz 2734, Valpara\'{\i}so, CHILE; email: hamdi.raissi@pucv.cl. This author  acknowledges the ANID funding Fondecyt 1160527.}
\footnote{The authors are grateful to Fulvio Pegoraro (Banque de France) for comments on the manuscript.}
\qquad \qquad
}
}

\begin{document}

    \maketitle


    \abstract{In this paper the orthogonal impulse response functions (OIRF) are studied in the non-standard, though quite common, case where the covariance of the error vector is not constant in time. The usual approach for taking into account such behavior of the covariance consists in  applying the standard tools to sub-periods of the whole sample. We underline that such a practice may lead to severe upward bias. We propose a new approach intended to give what we argue to be a more accurate resume of the time-varying OIRF. This consists in averaging the Cholesky decomposition of nonparametric  covariance estimators.  In addition an index is developed to evaluate the heteroscedasticity effect on the OIRF analysis. The asymptotic behavior of the different estimators considered in the paper is investigated. The theoretical results are illustrated by Monte Carlo experiments. The analysis of the orthogonal response functions of the U.S. inflation to an oil price shock, shows the relevance of the tools proposed herein for an appropriate analysis of economic variables.}


%
%
%
%

\quad


    \section{Introduction}
    \label{intro}

In time series econometrics, it is common to investigate sub-samples of a full time series \color{black} in order to capture changes in the data. Reference can be made to Dees and Saint-Guilhem (2011), or Diebold and Yilmaz (2014) who considered rolling windows. In order to accommodate possible regime switches, Bernanke and Mihov (1998a) constitute different sub-periods for measuring monetary policy. Strongin (1995) split the data considered for the study according to the Federal Reserve operating procedures. Nazlioglu, Soytas and Gupta (2015) propose a pre-crisis, in-crisis and post-crisis split type to carry out a volatility spillover analysis, while Strohsal, Proa\~{n}o and Wolters (2019) consider pre and post 1985 financial liberalization samples. Blanchard and Simon (2001), Stock and Watson (2005) and Alter and Beyer (2014), use both rolling windows and static periods, to describe non constant dynamics in the series they study.

Our main message, focused on the orthogonal impulse response functions (OIRF) analysis, is that if one wishes to work with fixed sub-samples (for periods comparisons), it is advisable one to carry out a pointwise estimation, and then resume it using averages according to the periods of interest. This leads us to introduce in the following what will be called the \emph{averaged OIRF}. By doing so, an accurate picture of the non constant dynamics is obtained. As a matter of fact, applying the standard tools to sub-samples can, in some sense, lead to bias distortions in resuming the time-varying dynamics of a series. Several available approaches for the pointwise OIRF estimation could be used for our task (see Primiceri (2005), or Giraitis, Kapetanios and Yates (2018)). In this paper, we develop the above-presented idea in the important case of vector autoregressive (VAR) models with constant autoregressive parameters but with time-varying covariance structure. Indeed, it is often admitted that the conditional mean is constant, while the variance is time-varying  (see Bernanke and Mihov (1998b), Sims (1999), Stock and Watson (2002), Kew and Harris (2009), Cavaliere, Rahbek and Taylor (2010) or Patilea and Ra\"{\i}ssi (2012) among others). In addition, it is widely known that non constant variance is common for economic variables. For instance Sensier and van Dijk (2004) found that more than 80\% of the 214 U.S. economic variables they studied have a non constant variance (see Aue, H\"{o}rmann, Horv\`{a}th and Reimherr (2009) for break detection in the covariance structure). For this reason, the series under study are assumed non stationary, due to the non constant variance (i.e. the heteroscedasticity is unconditional). In the following, we present a univariate example commonly considered in the literature, to illustrate the main idea of the paper.

\subsection{A univariate example}

Many economic variables display huge shifts or noticeable long-range effects. For instance, emerging countries have experienced fast changes during the period 1990-2020. Also, important changes can be observed in the commodities markets. Thus, let us consider the log differences of the monthly global price of brent crude, in U.S. Dollars per barrel, from January 1990 to July 2020\footnote{The data can be download from the website site of the federal reserve bank of Saint Louis, https://fred.stlouisfed.org/series/POILBREUSDM}. From Figure \ref{fig1}, it seems that the oil prices log differences are subject to clear heteroscedastic effects. Using the adaptive approach introduced by Xu and Phillips (2008), the conditional mean of the conditional mean is filtered by fitting an AR model.

Using this simple framework, we illustrate the ways of resuming the time-varying response functions to a \emph{rescaled}\footnote{The term \emph{rescaled} is taken from L\"{u}tkepohl (2005,p53).} impulse for an univariate series.
Let us define by $\sigma_t^2=g^2(t/T)$ the (unobserved) innovations variance at time $1\leq t \leq T$, where $g(.)$ is a function fulfilling some regularity conditions. As the (unobserved) moving average coefficients $\phi_i$ are constant in our case, it suffices to focus on the changes in the variance to capture the evolutions of the rescaled impulse response functions (IRF) $\phi_i\sigma_{t-i}$. The usual way to resume the time varying IRF over a given period would be to estimate the standard tool that assume a constant variance. This would lead to estimate  $\phi_i(\int g^2)^{0.5}$ (which will be called the approximate IRF), whereas $\phi_i\int g$ (which will be called  the averaged IRF) is more sound to resume the IRF. Here, the integrals account for the averaging over given periods of interest. Indeed, if the purpose is to find a resume of the IRF, averaging over the values of fixed periods seems more reasonable than considering a kind of norm such as $(\int g^2)^{0.5}$. Clearly the averaged and approximated IRF are in general different, as long as the variance structure is non constant. More precisely, more the variance varies over time, larger the discrepancy between the averaged and the approximated IRF is. Therefore, in the following we also  propose to build  and indicator of the variability of the variance based on the discrepancy between the averaged and approximated IRFs. It is important to underline that the robustness/stability studies often rely on the simple (graphical) examination of the different OIRF. As this way of proceeding is subjective, our indicator is intended to quantify such a kind of analyses.

In order to have an idea about the differences between the two quantities, kernel estimators of $\int g$ and $(\int g^2)^{0.5}$ are given for the pre and post 2008 subprime mortgage crisis (see Table \ref{tab00}). As expected, the approximated IRF estimates are noticeably greater than the averaged IRF estimates. For example, we found that the approximated IRF can be greater than the averaged IRF by 51.73\% for post-crisis period, and 4.09\% for the pre-crisis period. These discrepancies are consequences of fast changes in the variance of the series. This can be explained by the fact that the approximated IRF uses the squared observations (or residuals), while the averaged IRF considers estimates of the variance structure. In our particular example, comparing the pre and post crisis outputs, a practitioner could conclude that the IRFs did not significatively changed in mean examining the averaged estimators. If one relies on the approximated approach, then one could spuriously find an increase of the IRF due to the crisis effect.

The structure of this study is as follows. In section \ref{impulse} the vector autoregressive model with unconditionally heteroscedastic innovations is presented. Next, different possible concepts of OIRF that could be considered in our framework are discussed. Moreover, we introduce a scalar variance variability index that measures the departure from the standard constant variance VAR setup. The Section \ref{kernel_sec} is dedicated to the estimators and their asymptotic properties. The time-varying OIRF estimator is introduced and its nonparametric rate of convergence is derived. In Section \ref{ols} and Section \ref{adapti}, the estimators of the approximated and averaged OIRF are defined. Their asymptotic behavior is also studied. In Section \ref{vs_index}, we introduce the estimator of our variance variability index, and derive asymptotic properties. In Section \ref{numerical}, Monte Carlo experiments are conducted to compare the finite sample properties of the different estimators of the OIRF. Oil price and U.S. inflation variables are considered to underline the usefulness of the proposed tools. The proofs are relegated to the Appendix.

\section{Time-varying orthogonal impulse response functions}
\label{impulse}

Following the usual approach for impulse response analysis between variables, consider a vector autoregressive (VAR) model for the series $X_t\in\mathbb{R}^d$:
\begin{eqnarray}
&&{X}_t={A}_{01}{X}_{t-1}+\dots+{A}_{0p}{X}_{t-p}+u_t\label{hetero1}
\end{eqnarray}
where $u_t$ is the error term and the $A_{0i}$'s are the autoregressive parameters matrices, supposed to be  such that $\det {A}(z)\neq 0$ for
all $|z|\! \leq \! 1$, with ${A}(z) \!=\!
I_d-\sum_{i=1}^{p}\!{A}_{0i} z^i$.
Here, the covariance of the system is allowed to vary
in time.  More precisely, the covariance of the process $(u_t)$ is denoted by $\Sigma_t:=G(t/T)G(t/T)'$,  where $r\mapsto G(r),$ $r\in (0,1]$, is a $d\times d-$matrix valued function.
With the rescaling device used by Dahlhaus (1997), the
process $(X_t)$ should be formally written in a triangular form. Herein,  the double subscript is suppressed for notational simplicity.

The specification we consider allows for commonly
observed features as cycles, smooth or abrupt changes for the covariance, and is widely used in the literature (see e.g. Xu and Phillips (2008) and references therein). In particular, the rescaling device is commonly used to describe long-range phenomena (see Cavaliere  and Taylor (2007, 2008) among others).
In practice the lag length $p$ in (\ref{hetero1}) is unknown but can be fixed using the tools proposed in Patilea and Ra\"{\i}ssi (2013) and Ra\"{\i}ssi (2015) under our assumptions.

In the sequel, the model (\ref{hetero1}) is considered re-written as follows:
\begin{eqnarray*}
&& X_t=(\widetilde{X}_{t-1}'\otimes I_d)\vartheta_0+u_t\\&&
u_t=H_t\epsilon_t,
\end{eqnarray*}
where $(\epsilon_t)$ is an iid centered process with $E(\epsilon_t\epsilon_t^\prime)= I_d$ and $$\vartheta_0=\mbox{vec}(A_{01},\dots,A_{0p})$$ is the vector of parameters. Herein the vec($\cdot$) operator consists in stacking the columns of a matrix into a vector. The matrix $H_t$ is the lower triangular matrix of the
Cholesky decomposition of the errors' covariance, that is $\Sigma_t=H_tH_t'$. The matrix $I_d$ is the $d\times d-$identity matrix.
The usual Kronecker product is denoted by $\otimes$ and $\widetilde{X}_{t}=(X_{t}',\dots,X_{t-p+1}')'$. We also define 
\begin{equation}\label{impul}
\Phi_0=I_d,\quad\Phi_i=\sum_{j=1}^i\Phi_{i-j}A_{0j},
 \end{equation}
$i=1,\dots$, with $A_{0j}=0$ for $j>p$. The
$\Phi_i$'s correspond to the coefficients matrices of the infinite moving average representation of $(X_t)$. Under our assumptions the components of the $\Phi_i$'s decrease exponentially fast to zero.



If the errors' covariance $\Sigma$ is assumed constant, then we can define $d\times d$ standard OIRF
\begin{equation}\label{exameq}
\theta(i):=\Phi_iH,\qquad i=1,2,\dots
\end{equation}
where here $H$ is the lower triangular matrix of the
Cholesky decomposition of $\Sigma$. See L\"{u}tkepohl (2005, p59). Let us denote by $\widehat{\vartheta}_{OLS}$ the ordinary least squares (OLS) estimator of the autoregressive parameters
and define $\widehat{\Sigma}$ the OLS estimator of the constant errors covariance matrix.
Using $\widehat{\vartheta}_{OLS}$ and $\widehat{\Sigma}$, it is easy to see that an estimator of $\theta(i)$ can be built. Under standard assumptions, it can be shown that such estimators are consistent, $\sqrt{T}$-asymptotically Gaussian. See L\"{u}tkepohl (2005, p110). However, it clearly appears that the classical OIRF cannot take into account for the time-varying instantaneous effects properly, and may be misleading in our non standard but quite realistic framework.


\subsection{tv-OIRF}

In the framework of the model (\ref{hetero1}), a common alternative to the classical OIRF is
 the time-varying OIRF (tv-OIRF hereafter)
\begin{equation}\label{vartheta}
\theta_{r}(i):=\Phi_iH(r),\qquad i\geq 1,
\end{equation}
\emph{for each $r\in(0,1]$}, and where  $H(r)$ is the lower triangular matrix of the Cholesky decomposition of $\Sigma(r)=G(r)G(r)^\prime $. The parameter $r$ gives the time where the impulse response analysis is conducted. In other words, the counterpart of the usual OIRF in the case of time-varying variance is the two arguments function
$$(r,i)\mapsto \theta_{r}(i),\qquad (r,i]\in (0,1]\times \{1,2,\ldots\}.$$
The form (\ref{vartheta}) implicitly arises when models with constant autoregressive parameters but time varying variance are used to analyse the data (see Bernanke and Mihov (1998a), Stock and Watson (2002) or Xu and Phillips (2008) among others for this kind of models). When the covariance of the errors is constant, for each $i$ the map $r\mapsto \theta_{r}(i)$ is constant, and thus we retrieve the standard case.
Although it is interesting to have a pointwise estimation of the OIRF, in general these maps are not constant and are typically estimated at nonparametric rates, as it will be shown in the following. Primiceri (2005) and Giraitis, Kapetanios and Yates (2018) have provided complete tools for estimating (\ref{vartheta}) in general contexts. As a byproduct of the main results of the paper, we specify the methodological pathway for the pointwise estimation of the OIRF in the important case where the conditional mean is constant and the variance is time-varying.

Some resume of the tv-OIRF through time could be sometimes needed to compare fixed periods. In many cases, this consists in evaluating the differences between pre and post crises situations. In the sequel, we consider two approaches for resuming the tv-OIRF  over a given sub-period. First, we replace the matrix $H(r)$ in equation (\ref{vartheta}) by the lower triangular matrix of the Cholesky decomposition of the \emph{realized variance}, that is the average of the variance, over a given period around $r$. This will yield to what we shall call \emph{approximated OIRF}. Typically, this corresponds to the usual practice which consists in applying the standard method to periods (see e.g. Stock and Watson (2005) or Beetsma and Giuliodori (2012)). Second, keeping in mind that we are looking for a resume of the tv-OIRF, which is tantamount to looking for a resume of integrated $H(r)$ appearing in  equation (\ref{vartheta}), we introduce the \emph{averaged OIRF} that is obtained by replacing the matrix $H(r)$ with the average of the lower triangular matrix of the Cholesky decomposition of $\Sigma(\cdot)$ over a given period around $r$. Both resumes we consider could be estimated  at parametric rates and, considering static or rolling periods, could be used for an analysis of the series. However, as argued in the Introduction, the averaged OIRF should be preferred. Before presenting the approximate and averaged approaches, let us point out that, as usual, resuming the OIRF does not makes the shocks orthogonal pointwise. Note however that such a property is not really needed if we are interested in comparing periods by considering means. 

\subsection{Approximated OIRF}
\label{ols_def}

The usual way to resume the OIRF in presence of a non constant covariance in our framework is to consider the following quantities
\begin{equation}\label{approxorth}
\widetilde{\theta}_{r}^{q}(i)=\Phi_i\widetilde{H}(r),\qquad i\geq 1,
\end{equation}
where $\widetilde{H}(r)$ is the lower triangular matrix of the  Cholesky decomposition of the positive definite matrix $q^{-1}\int_{r-q/2}^{r+q/2}\Sigma(v)dv$ with $0<r-q/2<r+q/2<1$. Again the standard case is retrieved if the covariance structure is assumed constant. If $r$ does not corresponds to a covariance break, we have $\widetilde{\theta}_{r}^{q}(i)\approx\theta_{r}(i)$ for small enough $q$. However, as the periods under study are usually somewhat large, so we are not aiming in reflecting the evolutions of $H(\cdot)$, we will refer to (\ref{approxorth}) as approximate OIRF in the sequel. In short, the approximated OIRF are usually computed to contrast between static periods. For fixed $r$ and $q$, the quantities  $\widetilde{\theta}_{r}^{q}(i)$ could be estimated at parametric rates.

\subsection{Averaged OIRF}
\label{adapti_def}

As argued above, by construction, the approximated OIRF could be misleading in resuming the time-varying $\theta_{r}(i)$ over a period. Given the definition of $\theta_{r}(i)$, a more natural way to approach it would be to average the lower triangular matrix of  the Cholesky decomposition over a window around $r$. We propose a new alternative way to resume the tv-OIRF (\ref{vartheta}) based on the quantities
\begin{equation}\label{imprespint}
\bar{\theta}_{r}^{q}(i):=\Phi_i\bar{H}(r)\:\mbox{where}\:\bar{H}(r):=\frac{1}{q}\int_{r-q/2}^{r+q/2}H(v)dv,\qquad i\geq 1,
\end{equation}
$0<q<1$ is fixed by the practitioner, and $r$ is such that $0<r-q/2<r+q/2<1$.
The standard case is retrieved if the errors covariance is assumed constant. On the other hand if $r$ does not correspond to an abrupt break of the covariance structure, we clearly have $\bar{\theta}_{r}^{q}(i)\approx\theta_{r}(i)$ when $q$ is small. However, as noted above, the averaged OIRF is intended to be applied for a relatively large $q$.

\subsection{Variance variability indices}
\label{index}

In this section we propose an index, that is a scalar, to measure the departure from a constant covariance matrix situation within a given period.
%
%
%
We could write
$$
 \widetilde{\theta}_{r}^{q}(i) = \bar{\theta}_{r}^{q}(i)  I_{r,q}
$$
with
\begin{equation}
\label{index_var_matrix}
I_{r,q} =  \bar H(r)^{-1}\widetilde H(r).
\end{equation}
 Let us define
\begin{equation}
\label{index_var_matrix_b}
i_{r,q} = \color{black}  \left\| I_{r,q} \right\|_2^2 , \color{black}
\end{equation}
where here $\|\cdot\|_2$ denotes the spectral norm of a matrix. In this case, $i_{r,q}$ is equal \color{black} to the square of \color{black} the largest eigenvalue of $I_{r,q}$, which has only real, positive eigenvalues.  By elementary matrix algebra properties, we also have
$$
i_{r,q} =
\max_{a\in\mathbb{R}^d,a\neq 0 } \frac{Var (a^\prime X_{app})}{Var (a^\prime X_{avg})},
$$
where $X_{avg}$ and $X_{app}$ are $d$-dimensional random vectors with variances $\bar H(r) \bar H(r) ^\prime$ and $\widetilde H(r)\widetilde H(r)^\prime $, respectively. In the statistical literature, a quantity like $i_{r,q}$ is usually called the first relative eigenvalue of one matrix (here $\widetilde H(r) \widetilde H(r)^\prime$) with respect to the other matrix (here $\bar H(r)  \bar H(r) ^\prime$).  See, for instance, Flury (1985). By construction, in our context, the eigenvalues of the matrix $\bar H(r)^{-1}\widetilde H(r)$ are   real numbers   larger than or equal to 1, as shown in the following.

The index $i_{r,q}$ is inspired by the OIRF analysis. It is designed to provide a measure of variability through the contrast between two possible definitions of OIRF that coincide in the case of a covariance $\Sigma$ constant over time. Another simple index could be defined as
\begin{equation}\label{index_var_matrix_c}
j_{r,q} =  \left\| \frac{1}{q}\int_{r-q/2}^{r+q/2}\Sigma(v)dv - \bar H(r)\bar H(r)^\prime   \right\|_2 ^{  2 }.
\end{equation}
 By elementary properties of the spectral norm,
$$
j_{r,q}  = \max_{a\in\mathbb{R}^d,\| a \| = 1 } \{Var (a^\prime X_{app}) - Var (a^\prime X_{avg})\}.
$$
The technical assumptions of the paper may be found in Section \ref{sec.assump}.
\begin{lemma}\label{CS_ineq_m}
Under the Assumption \textbf{A1},
\begin{enumerate}
\item $i_{r,q}\geq 1$ and  $i_{r,q}=1$ if and only if   $v\mapsto H(v)$ is constant on $(r-q/2, r+q/2)$;
\item    $j_{r,q}\geq 0$ and \color{black} $j_{r,q}=0$ if and only if   $v\mapsto H(v)$ is constant on $(r-q/2, r+q/2)$
\end{enumerate}
\end{lemma}

In our context, for any $0<q<1$, a map $r\mapsto i_{r,q}$ (resp. $r\mapsto j_{r,q}$) constant equal to 1 (resp. 0) means the covariance of $X_t$ is constant in time. For simplicity, in the sequel we will focus on index $ i_{r,q}$ which is invariant to multiplication of the errors' covariance matrix by a positive constant.  Large  values of $i_{r,q}$  indicates a large variability in the variance of the vector series in a given period of interest.

\section{OIRF estimates when the variance is varying}\label{kernel_sec}
\setcounter{equation}{0}

Let us first briefly recall the estimation methodology for heteroscedastic VAR models of Patilea and Ra\"{\i}ssi (2012), Section 4. 
{First, we consider the OLS estimator of the autoregressive parameters}

\begin{equation}\label{estOLS}
{\widehat{\vartheta}_{OLS}=\left\{\sum_{t=1}^{T}\widetilde{X}_{t-1}\widetilde{X}_{t-1}'\otimes I_d\right\}^{-1}\mbox{vec}\left(\sum_{t=1}^{T} X_t\widetilde{X}_{t-1}'\right).}
\end{equation}
\color{black} {Patilea and Ra\"{\i}ssi (2012) showed that }
\begin{equation}\label{res2ols}
{\sqrt{T} (\widehat{\vartheta}_{OLS}-\vartheta_0)\Rightarrow
\mathcal{N}(0,\Lambda_3^{-1}\Lambda_2\Lambda_3^{-1}),}
\end{equation}
where
\begin{equation}\label{lamb2_3}
{\Lambda_2=\!\int_0^1
\!\sum_{i=0}^{\infty}\!\left\{\!\tilde{\Phi}_i(\mathbf{1}_{p\times
p}\otimes\Sigma(r))\tilde{\Phi}_i'\right\}\!\otimes\!\Sigma(r)dr,}
\quad
{\Lambda_3=\! \int_0^1 \!
\sum_{i=0}^{\infty}\!\left\{\!\tilde{\Phi}_i(\mathbf{1}_{p\times
p}\otimes\Sigma(r))\tilde{\Phi}_i'\right\}\!\otimes\! I_d\:dr,}
\end{equation}
with $\mathbf{1}_{p\times p}$ the $p\times p$ matrix with
components equal to one, and $\widetilde{\Phi}_{i}$ is a block diagonal matrix $\widetilde{\Phi}_{i}:=diag\left(\Phi_i,\Phi_{i-1},\dots,\Phi_{i-p+1}\right)$.
The matrices $\Phi_i$, $i\geq 0$, are defined in equation (\ref{impul}), and $\Phi_i=0$ for $i<0$.

Next, let us consider kernel estimators of the time-varying covariance matrix. \color{black}
Denote by $A\odot B$ the Hadamard  (entrywise) product of two matrices of same dimension $A$ and $B$. For $t=1,\dots,T$, define the symmetric matrices
\begin{equation}\label{Sig_NP}
\widehat{\Sigma}_t=\sum_{j=1}^{T} w_{tj} \odot \widehat{u}_j\widehat{u}_j^\prime ,
\end{equation}
where the $\widehat{u}_t = X_t - (\widetilde{X}_{t-1}^\prime \otimes I_d) \widehat{\vartheta}_{OLS}$ are the OLS residuals. The $(k,l)-$element, $k\leq l$, of the $d\times d$ matrix of weights $w_{tj}$ is given by
$$w_{tj}(b_{kl})= (Tb_{kl})^{-1} K\left((t-j)/(Tb_{kl})\right),$$
with $b_{kl}$ the bandwidth and  $K(\cdot)$ a nonnegative kernel function. \color{black} For any $r\in(0,1]$, the value $\Sigma(r)$ of the covariance function could be estimated by $\widehat{\Sigma}_{[rT]}$. (Here and in the following, for a number $a$, we denote by $[a]$ the integer part of $a$, that is the largest integer number smaller or equal to $a$.)
For all $1\leq k\leq l\leq d$ the bandwidth $b_{kl}$ belongs to a range $\mathcal{B}_T = [c_{min} b_T, c_{max} b_T]$ with $c_{min}, c_{max}>0$ some constants and $b_T \downarrow 0$ at a suitable rate specified below. In practice the bandwidths $b_{kl}$ can be chosen by minimization of a cross-validation criterion. \color{black} This estimator is a version of the Nadaraya-Watson estimator considered by Patilea and Ra\"{\i}ssi (2012). Here, we replace the denominator by the target density, that is the uniform density on the unit interval which is constant equal to 1.  \color{black} A regularization term may be needed to ensure that the matrices $\widehat{\Sigma}_t$ are positive definite (see Patilea and Ra\"{\i}ssi (2012)). Another simple way to circumvent the problem is to select a unique bandwidth $b=b_{kl}$, for all $1\leq k,l\leq d$.

With at hand an estimator of  $\Sigma(r)$, we could define $\widehat{H}_{[rT]}$, the lower triangular matrix of the Cholesky decomposition of $\widehat{\Sigma}_{[rT]}$, as the estimator of $H(r)$. Below, we establish the convergence rates of these nonparametric estimates. For $r\in(0,1)$, let $\Sigma(r-) = \lim_{\tilde r \uparrow r} \Sigma (\tilde r)$ and  $\Sigma(r+) = \lim_{\tilde r \downarrow r} \Sigma (\tilde r)$. Moreover, by definition let $\Sigma(1+)=0$. Let $H(r-)$ and  $H(r+)$ be defined similarly.
In the following, $\|\cdot\|_F$ is the Frobenius norm, while $\sup_{\mathcal{B}_T}$ denotes the supremum with respect the bandwidths $b_{kl}$ in $\mathcal{B}_T$.

\vspace{0.3 cm}

\begin{proposition}\label{prop_tv_oirf}
Assume that  Assumptions {\bf A0}-{\bf A2} in the Appendix hold true. Then, for any $r\in(0,1]$,
$$
\sup_{\mathcal{B}_T} \left\| \widehat{\Sigma}_{[Tr]} - \frac{1}{2} \left\{\Sigma(r-) + \Sigma(r+)\right\}\right\|_F  = O_{\mathbb{P}}\left( b_T+\sqrt{\log(T)/Tb_T} \right)
$$
and
$$
\sup_{\mathcal{B}_T} \left\| \widehat{H}_{[Tr]} -  \frac{1}{2}  H_{\pm}(r)  \right\|_F  = O_{\mathbb{P}}\left( b_T+\sqrt{\log(T)/Tb_T} \right),
$$
where $H_{\pm}(r)$ is the lower triangular matrix of the Cholesky decomposition of $\Sigma(r-) + \Sigma(r+)$.
\end{proposition}

\vspace{0.3 cm}

The convergence rate of $\widehat{\Sigma}_{[Tr]}$ and $\widehat{H}_{[Tr]}$ is given by a bias term, with the standard rate one could obtain when estimating Lipschitz continuous functions nonparametrically, and a variance term which is multiplied by a logarithm factor, the price to pay for the uniformity with respect to the bandwidth.

\color{black}

The above estimation of the non constant covariance structure \color{black} could be \color{black}  used to define the adaptive least squares (ALS) estimator

\begin{equation}\label{GLS}
\widehat{\vartheta}_{ALS}=\widetilde{\Sigma}_{\widetilde{\underline{X}}}
^{-1}\mbox{vec}\:\left(\widetilde{\Sigma}_{\underline{X}}\right),
\end{equation}
where
$$\widetilde{\Sigma}_{\widetilde{\underline{X}}}=T^{-1}\sum_{t=1}^{T}\widetilde{X}_{t-1}\widetilde{X}_{t-1}'
\otimes\widehat{\Sigma}_t^{-1}
\quad\mbox{and}\quad\widetilde{\Sigma}_{\underline{X}}=T^{-1}
\sum_{t=1}^{T}\widehat{\Sigma}_t^{-1}X_t\widetilde{X}_{t-1}'.$$
\color{black} By minor adaptation of the proofs in Patilea and Ra\"{\i}ssi (2012), in order to take into account the simplified change in the definition of the weights $w_{tj} $, it could be shown that, uniformly with respect to $b\in\mathcal{B}_T$,  $\widehat{\vartheta}_{ALS}$ is consistent in probability and \color{black}
\begin{equation*}
{\sqrt{T}(\widehat{\vartheta}_{ALS}-\vartheta_0)\Rightarrow
\mathcal{N}(0,\Lambda_1^{-1}),}
\end{equation*}
where
\begin{equation}\label{lambda_1}
{\Lambda_1=\int_0^1
\sum_{i=0}^{\infty}\left\{\widetilde{\Phi}_i(\mathbf{1}_{p\times
p}\otimes\Sigma(r))\widetilde{\Phi}_i'\right\}\otimes\Sigma(r)^{-1}dr}.
\end{equation}
  Patilea and Ra\"{\i}ssi (2012) showed that $\Lambda_3^{-1}\Lambda_2\Lambda_3^{-1} - \Lambda_1 $ is a positive semi-definite matrix.
 \color{black}

\subsection{The tv-OIRF nonparametric estimator}

In the context of model (\ref{hetero1}), the natural way to build estimates of the time-varying OIRF defined in equation (\ref{vartheta}) is to plugin estimates of the $\Phi_i$ and $H(r)$. For estimating $\Phi_i$  we use $\widehat{\Phi}_i^{als}$ which are obtained as in (\ref{impul}), but considering the ALS estimator of the $A_{0i}$'s. By the arguments used in the proof of Proposition \ref{resapprox} below, this estimator has the $O_{\mathbb{P}}(1/\sqrt{T})$ rate of convergence. Using the nonparametric estimator of $H(r)$ we introduced above, we obtain what we will call the \emph{ALS estimator of} $\theta_r(i)$, that is
\begin{equation}\label{av1als1}
\widehat{\theta}_{r}(i):=\widehat{\Phi}_i^{als}\widehat{H}_{[rT]},\qquad r\in(0,1].
\end{equation}
Even if $\widehat{\Phi}_i^{als}$ has an improved variance compared to the estimator one would obtain using the OLS
estimator of the $A_{0i}$'s, the estimator $\widehat{\theta}_{r}(i)$ still inherits the nonparametric rate of convergence of   $\widehat{H}_{[rT]}$ described in Proposition \ref{prop_tv_oirf}. Hence, analyzing the variations of the estimated curves $r\mapsto\widehat \theta_{r}(i)$, for various $i$,  suffers from lower, nonparametric convergence rates.  In section \ref{adapti} we propose to use instead of $\widehat \theta_{r}(i)$ averages over the values in a neighborhood of $r$, that is a window containing $r$. In particular, this allows to recover parametric rates of convergence of the estimators. In practice, this interval correspond to some period of interest.

\subsection{Approximated orthogonal impulse response function estimates}
\label{ols}

The results of this part are only stated as they are direct consequences of arguments in
\color{black}
Patilea and Ra\"{\i}ssi (2012)
\color{black}
and standard techniques (see L\"{u}tkepohl (2005)). Let the usual estimator of (\ref{approxorth}),
\begin{equation}\label{olsapproxhist}
\hat{\widetilde{\theta}}_{r}^{q}(i):=\widehat{\Phi}_i^{ols}  
 \widehat{\widetilde H}(r), \color{black}
\end{equation}
where $\widehat{\widetilde H}(r)$ is the lower triangular matrices of the Cholesky decomposition of
\begin{equation}\label{H_trad2}
\widehat S_T (r) = \frac{1}{[qT]+1}\sum_{k=- [qT/2]}^{[qT/2]}\widehat{u}_{[rT]-k}\widehat{u}_{[rT]-k}',
\end{equation}
with $\widehat{u}_{[rT]-k}$ the OLS residuals  {and $\widehat{\Phi}_i^{ols}$ are the estimators of the MA coefficients obtained from the OLS estimators of the autoregressive parameters.}
Recall that (\ref{olsapproxhist}) is used to evaluate the OIRF in the standard homoscedastic case (see L\"{u}tkepohl (2005) Section 3.7), but is also commonly considered to evaluate tv-OIRF in static periods. The expression (\ref{olsapproxhist}) is suitable at least asymptotically, \color{black} since by the proof Lemma \ref{lemma_aoirf} below
\begin{multline}\label{H_trad}
\frac{1}{[qT]+1}\sum_{k=- [qT/2]}^{[qT/2]}\widehat{u}_{[rT]-k}\widehat{u}_{[rT]-k}'
=  
 \frac{1}{[qT]+1}\sum_{k=- [qT/2]}^{[qT/2]}{u}_{[rT]-k}{u}_{[rT]-k}' +o_{\mathbb{P}}(1/\sqrt{T})
\\ =\frac{1}{q}\int_{r-q/2}^{r+q/2}\Sigma(v)dv
  +O_{\mathbb{P}}(1/\sqrt{T}).
\end{multline}


In order to specify the asymptotic behavior of $\hat{\widetilde{\theta}}_{r}^{q}(i)$, we first state a result which can be proved using similar arguments to those of Lemma 7.4 of Patilea and Ra\"{\i}ssi (2010). Let $\widehat{\zeta}_t:=\mbox{vech}\left(\widehat{u}_t\widehat{u}_t'\right)$, $\zeta_t:=\mbox{vech}\left(u_tu_t'\right)$
 and $\Gamma_t:=\mbox{vech}(\Sigma(t/T))=\mbox{vech}(\Sigma_t)$, where the vech operator consists in stacking the elements on and below the main diagonal of a square matrix. Define
$\overline{\Gamma}(r):=\mbox{vech}\left(\color{black} q^{-1} \color{black} \int_{r-q/2}^{r+q/2}\Sigma(v)dv\right)$
  and $\widehat{\overline{\Gamma}}(r):=([qT]+1)^{-1}\sum_{k=-[qT/2]}^{[qT/2]}\widehat{\zeta}_{[rT]-k}$ for $r<1$. \color{black} Introduce also
the functions $\Gamma(\cdot)$ and $\Delta(\cdot)$ given by $\Gamma(\cdot)=\mbox{vech}(\Sigma(\cdot))$ and $\Delta(t/T)=E(\zeta_t\zeta_t')$.

\vspace{0.3 cm}

\begin{lemma}\label{lemma_aoirf}
Under the assumptions {\bf A0}-{\bf A3} in the Appendix,
we have
\begin{equation}\label{res0}
\color{black} \sqrt{T} \color{black} \left(
                 \begin{array}{c}
                   \color{black} \widehat{\vartheta}_{OLS} \color{black} -\vartheta_0 \\
                   \widehat{\overline{\Gamma}}(r)-\overline{\Gamma}(r) \\
                 \end{array}
               \right)\Rightarrow
\mathcal{N}\left(0,\left(
                \begin{array}{cc}
                  \Lambda_3^{-1}\Lambda_2\Lambda_3^{-1} & 0 \\
                  0 & \Omega(r) \\
                \end{array}
              \right)\right),
\end{equation}
with $\widehat{\vartheta}_{OLS}$  defined in \eqref{estOLS}, $\Lambda_2$, $\Lambda_3$ defined in \eqref{lamb2_3} and
$$\Omega(r)=\frac{1}{q}\int_{r-q/2}^{r+q/2}\{\Delta(v)-\Gamma(v)\Gamma(v)' \} dv.$$
\end{lemma}

\vspace{0.3 cm}

Now, define the commutation matrix $K_d$ such that $K_d\mbox{vec}(G)=\mbox{vec}(G')$, and the elimination matrix $L_d$ such that $\mbox{vech}(G)=L_d\mbox{vec}(G)$ for any square matrix $G$ of dimension $d\times d$. 
{Introduce the $pd\times pd$ matrix}
\begin{equation}\label{calA}
{\mathbb{A}=\left(
               \begin{array}{cccc}
                 A_{01} & \dots & \dots & A_{0p} \\
                 I_d & 0 & \dots & 0 \\
                 0 & \ddots & 0 & \vdots \\
                 0 & 0 & I_d & 0 \\
               \end{array}
             \right)}
\end{equation}
{and the $d\times pd$-dimensional matrix $J=(I_d,0,\dots,0)$}. We are in position to state the asymptotic behavior of the classical approximated OIRF estimator. Note that this result can be obtained using the same arguments of L\"{u}tkepohl (2005), Proposition 3.6, together with (\ref{res0}).

\vspace{0.3 cm}

\begin{proposition}\label{resapprox}
Under the Assumptions {\bf A0}-{\bf A1} in the Appendix,
we have for all $r\in(q/2,1-q/2)$ and as $T\to\infty$
\begin{equation}\label{histapproxasymp}
\sqrt{T}\mbox{vec}\left(\hat{\widetilde{\theta}}_{r}^{q}(i)-\widetilde{\theta}_{r}^{q}(i)\right)
\Rightarrow\mathcal{N}\left(0,C_i(r)\Lambda_3^{-1}\Lambda_2\Lambda_3^{-1}C_i(r)'+D_i(r)\Omega(r)D_i(r)'\right),\:i=0,1,2,...
\end{equation}
where $C_0=0$, $C_i(r)=\left(\widetilde{H}(r)'\otimes I_d\right)\left(\sum_{m=0}^{i-1}J(\mathbb{A}')^{i-1-m}\otimes\Phi_m\right)$, $i=1,2,...$, $\widetilde{H}(r)$ is given in (\ref{approxorth}),
and
$$D_i(r)=\left(I_d\otimes\Phi_i\right)\Xi(r), \:i=0,1,2,...$$
with
$$\Xi(r)=L_d'\left[L_d\left(I_{d^2}+K_{d}\right)\left(\widetilde{H}(r)\otimes I_d\right)L_d'\right]^{-1}.$$
\end{proposition}

\vspace{0.3 cm}

\color{black}

We propose an alternative approximated OIRF estimator based on the more efficient estimator $\widehat{\vartheta}_{ALS}$ defined in equation (\ref{GLS}) and the estimators  $\widehat{\Phi}_i^{als}$ of the coefficients $\widehat{\Phi}_i$ of the infinite moving average representation of $(X_t)$. More precisely,
\begin{equation}\label{ALSapproxhist}
  \hat{\widetilde{\theta}}_{r}^{q,als}(i):=\widehat{\Phi}_i^{als} \widehat{\widetilde H}(r),
\end{equation}
a new  approximated OIRF estimator. Below, we state its asymptotic distribution.

\begin{proposition}\label{resapproxALS}
Let the conditions of Proposition \ref{resapprox}, and the Assumption  {\bf A2} in the Appendix hold true. With the notation defined in Proposition \ref{resapprox},
we have for all $r\in(q/2,1-q/2)$ and as $T\to\infty$,
\begin{equation}\label{histapproxasymp_alt}
\sqrt{T}\mbox{vec}\left(\hat{\widetilde{\theta}}_{r}^{q, als}(i)-\widetilde{\theta}_{r}^{q}(i)\right)
\Rightarrow\mathcal{N}\left(0,C_i(r)\Lambda_1^{-1}C_i(r)'+D_i(r)\Omega(r)D_i(r)'\right),\:i=0,1,2,...
\end{equation}
Moreover, the difference between the asymptotic variance of $\mbox{vec}\left(\hat{\widetilde{\theta}}_{r}^{q}(i)-\widetilde{\theta}_{r}^{q}(i)\right)$ given in  equation \ref{histapproxasymp_alt} and the asymptotic variance of $\mbox{vec}\left(\hat{\widetilde{\theta}}_{r}^{q, als}(i)-\widetilde{\theta}_{r}^{q}(i)\right)$ is a positive semi-definite matrix.
\end{proposition}

The proof of Proposition \ref{resapproxALS} is omitted since it follows the steps of the proof of Proposition \ref{resapprox}, and use the results of Patilea and Ra\"{\i}ssi (2012) on the convergence in law of $\widehat{\vartheta}_{ALS}$. In particular, they proved that $\Lambda_3^{-1}\Lambda_2\Lambda_3^{-1} - \Lambda_1^{-1}$ is a positive semi-definite matrix and this implies that $\hat{\widetilde{\theta}}_{r}^{q, als}(i)$ is a lower variance estimator of $\widetilde{\theta}_{r}^{q}(i)$.

\color{black}

\vspace{0.3 cm}

Although the standard $\hat{\widetilde{\theta}}_{r}^{q}(i)$, or the more efficient estimator
$\hat{\widetilde{\theta}}_{r}^{q, als}(i)$ are easy to compute, for the reasons we detailed above, we believe that they are not appropriate tools to resume \color{black} the evolution of the tv-OIRF (\ref{vartheta}). Instead, we propose to use an estimator of the averaged OIRF. To build such an estimate of the averaged OIRF with negligible bias, we need a slightly modified kernel estimator of $\Sigma(\cdot)$ that we introduce in the next section. \color{black}

\subsection{New OIRF estimators with time-varying variance}
\label{adapti}

In this section, we propose an alternative estimator for the approximated OIRF and an estimator for the averaged OIRF we introduced in section \ref{adapti_def}.
To guarantee $\sqrt{T}-$asymptotic normality for these estimators,  we implicitly need suitable estimators of integral functionals under the form
$$
\frac{1}{q}\int_{r-q/2}^{r+q/2} A(v)   \Sigma (v)  dv
$$
with $A(\cdot)$ some given matrix-valued function. The estimator of such integral, obtained by plugging in the nonparametric estimator of the covariance structure introduced in equation (\ref{Sig_NP}), would not be appropriate as it suffers from boundary effects. More details on this problem are provided in
 section \ref{NP_integrals} in the Appendix. Therefore, in the sequel, we construct alternative bias corrected estimators for such integral functionals.

For  $-[(q+h)T/2]\leq k \leq [( q+h)T/2]$, we define \color{black}
\begin{equation}\label{nnapprx2}
{\widehat V}_{[rT]-k} =   
\frac{1}{T}
\sum_{j=[(r-(q-h)
/2)T]+1}^{[(r+ (q-h)
/2)T]} \color{black} \frac{1}{h} L\left( \frac{[rT]-k -j}{hT} \right)   \widehat{u}_j\widehat{u}_j'.
\end{equation}
Hereafter, for simplicity, we use the same bandwidth $h$ for all the $d^2$ components of the estimated matrix-valued integrals. Note that   ${\widehat V}_{[rT]-k}$ \color{black} is an estimator of $\Sigma_{[rT]-k}$.
Next, let  $\widehat{ H}_{[rT]-k}$ denote the lower triangular matrix of the
Cholesky decomposition of ${\widehat V}_{[rT]-k} $, that is
\begin{equation}\label{H_V}
{\widehat V}_{[rT]-k} = \widehat{ H}_{[rT]-k }\widehat{ H}_{[rT]-k }^\prime.
\end{equation}

We propose the following adaptive least squares estimators of the time-varying averaged OIRF:
\begin{equation}\label{accumulatedgls}
 \hat{\bar{\theta}}_{r}^{q}(i)= \widehat{\Phi}_i^{als}\; \bar{\widehat{H}}(r),
\end{equation}
where
\begin{equation}\label{accumulatedgls_H}
  \bar{\widehat{H}}(r)=   \frac{1}{[ q T]+1}  \sum_{k=-[(q + h)T/2]}^{[(q+ h)T/2]}\widehat{H}_{[rT]-k}.
\end{equation}

\begin{proposition}\label{propostu}
If assumptions {\bf A0}-{\bf A2} hold true,  then for all $r\in(q/2,1-q/2)$ and as $T\to\infty$,
\begin{equation*}
\sqrt{T}\mbox{vec}\left(\hat{\bar{\theta}}_{r}^{q}(i)-   \bar{\theta}_r^{q}(i)   \right)
\Rightarrow\mathcal{N}(0,\overline{C}_i(r)\Lambda_1^{-1}\overline{C}_i(r)' + D_i(r)\Omega(r)D_i(r)'), \quad \:i=0,1,2,...
\end{equation*}
with $D_i(r)$ and $\Omega(r)$  defined in Proposition \ref{resapprox}
and
$$\overline{C}_i(r)=\left(\frac{1}{q}\int_{r-q/2}^{r+q/2}H(v)'dv\otimes I_d\right)\left(\sum_{m=0}^{i-1}J(\mathbb{A}')^{i-1-m}\otimes\Phi_m\right).$$
\end{proposition}

\subsection{Estimation of the variance variability index}\label{vs_index}

Finally, we  build estimators for the variance variability index introduced in section \ref{index}. In the proof of Proposition \ref{propostu}, it is shown that the estimator $ \bar{\widehat{H}}(r)$
defined in equation (\ref{accumulatedgls_H}) behaves $\sqrt{T}-$asymptotically normal centered at
$\bar{H}(r):=\frac{1}{q}\int_{r-q/2}^{r+q/2}H(v)dv$.
Then, the estimator of the index $i_{r,q}$ is
\begin{equation}
\label{index_var_matrix_b_est}
\widehat i_{r,q} =
\left\| \bar{\widehat {H}}(r)^{-1}\; \widehat{\widetilde H}(r) \right\|_2^2 ,
\qquad 0< r-q/2 < r+q/2< 1,
\end{equation}
where $ \bar{\widehat{H}}(r)$ is defined in (\ref{accumulatedgls_H}), and  $\widehat{\widetilde H}(r)$ the lower triangular matrix of the Cholesky decomposition of $\widehat S_T (r)$ defined as in equation (\ref{H_trad2}).

\begin{proposition}\label{asy_index}
 Let assumptions {\bf A0}-{\bf A2} hold true. Let $0<q<1/2$ and $r\in(q/2,1-q/2)$. If  $i_{r,q}>1$, and all other eigenvalues of
the matrix ${\widetilde H} (r) ^\prime \bar { H}(r) ^{-1\prime }  \bar { H}(r) ^{-1} {\widetilde H} (r) $ are strictly smaller than $i_{r,q}$, then
$
\sqrt{T}  \left(\widehat i_{r,q} - i_{r,q} \right)
$
converges in distribution to a centered normal variable. If  $i_{r,q}=1$,  then
$
\widehat i_{r,q} - 1 = o_{\mathbb P}(1/\sqrt{T}).
$
\end{proposition}

\medskip

The estimator $\widehat i_{r,q} $ has a non standard rate of convergence in the case of constant variance $\Sigma(\cdot)$. Determining this rate and its limit in distribution remains an open problem to be studied in the future.

\section{Numerical illustrations}
\label{numerical}

Several contributions in the literature have documented potential problems for the statistical analysis or the interpretation of the OIRF.
For instance Benkwitz \textit{et al.} (2000) pointed out several issues related to the building of bootstrap confidence intervals (see also L\"{u}tkepohl \textit{et al.} (2015) and references therein for recent developments in this field). Furthermore, we refer to L\"{u}tkepohl (2005), Section 2.3, for a discussion on the problems of the variables ordering or the missing of relevant variables. In order to address these issues, numerous settings
were proposed in the literature.
Such interesting topics deserve a complete work in our framework, and are beyond the scope of this article.
Hence, our numerical outputs will focus on the OIRF estimation, and the finite sample behavior of the heteroscedasticity index $i_{r,q}$ introduced above. In particular, for the approximated OIRF approach, we will consider the estimator (\ref{ALSapproxhist}) which benefits from the more accurate ALS estimation in comparison to the classical estimator given in (\ref{olsapproxhist}). The approximated OIRF estimator will be compared to the averaged OIRF estimator (\ref{accumulatedgls}).

\subsection{Monte Carlo experiments}

In this part, the $\bar{\widehat {H}}(r)$ will be computed using two bandwidths, $h_1=\frac{q}{2\sqrt{3}}T^{-1/3}$ and $h_2=\frac{q}{2\sqrt{3}}T^{-2/7}$, to illustrate the effect of the bandwidth choice on the OIRF analysis. The constant $q/2\sqrt{3}$ corresponds to the standard deviation of a uniform distribution on an interval of length $q$, while the rates $T^{-1/3}$ and $T^{-2/7}$ are two possible theoretical choices.
In each experiment, $1000$ independent trajectories of the following bivariate VAR(1) system are simulated

\begin{equation}\label{simulDGP}
X_t=AX_{t-1}+u_t,\quad u_t=H_t\epsilon_t,
\end{equation}
where
$$A=\left(
      \begin{array}{cc}
        0.5 & -0.3 \\
        0.1 & 0.3 \\
      \end{array}
    \right),
$$
and the $\epsilon_t$'s are standard Gaussian iid. The covariance of the errors terms $\Sigma_t:=H_tH_t'$ is driven by a matrix of functions $\Sigma_t=\Sigma(t/T)$ with

$$\Sigma(r)=\left(
              \begin{array}{cc}
                \sigma^2_{11}(r) & \sigma_{12}(r) \\
                \sigma_{21}(r) & \sigma^2_{22}(r) \\
              \end{array}
            \right),
$$
where $\sigma^2_{11}(r)=1.4+\delta f(r)$ for a fixed non constant function $f(\cdot)$ and $\delta\geq0$. $\sigma^2_{11}(r)$ is plotted in Figure \ref{first-var-comp}. The others components of the covariance matrix are set as follows: $\sigma^2_{22}(r)=0.5\sigma^2_{11}(r)$ and $\sigma_{12}(r)=\sigma_{21}(r)=\sqrt{\sigma^2_{11}(r)\sigma^2_{22}(r)}\times0.7$. The patterns displayed by the covariance structure are intended to mimic business cycle behavior commonly observed for economic variables. Note that when $\delta=0$, we retrieve the homoscedastic case. Samples $T=100,200,400$ and $800$ are considered in the sequel.

In the Monte Carlo investigation, the changes through time are studied by considering the subsample $(0.5;0.5)$, that is taking $q=0.5$ and $r=0.5$ (i.e. $i_{0.5,0.5}$). In order to avoid lengthy outputs, we only display the results for the orthogonalized response of the first variable for an impulse from its own past taking $i=1$. The corresponding averaged (resp. approximated) OIRF will be denoted by $\bar{\theta}_{0.5}^{0.5,11}(1)$ (resp. $\widetilde{\theta}_{0.5}^{0.5,11}(1)$).\\

We begin with a comparison between the averaged and approximated approaches for resuming the OIRFs. All the outputs concerning the OIRF are obtained setting $\delta=1$. In Figure \ref{fig2}, the relative differences between the averaged and approximated estimators are displayed. It appears that the approximated OIRF are in the order of 10\% greater than the averaged OIRF. The ratio is even always positive for $T=400$ and $T=800$. Recall that the approximated approach does not rely on the adequate way to resume the Cholesky decompositions of the covariance. Hence, we can conclude that the approximated approach delivers an upwards distorted picture of the OIRF when compared to the averaged approach. Now, let us turn to the illustration of the asymptotic results in Proposition \ref{resapproxALS} and \ref{propostu}. From the Q-Q plots displayed in Figure \ref{fig3} and \ref{fig4}, we can remark that the different OIRF estimates seems to behave as normal, even for small samples. In particular, we did not notice major differences between the estimators of the averaged OIRF obtained using the bandwidths $h_1$ and $h_2$.

In this part, we analyze the finite sample behavior of the index estimator defined in (\ref{index_var_matrix_b_est}). Recall that the index is intended to capture the discrepancy, between the homoscedastic and the heteroscedatic cases. Figure \ref{fig5} and \ref{fig6} correspond to a heteroscedasticity parameter $\delta=1$. In Figure \ref{fig7}, various values are considered for $\delta$, meanwhile the outputs for the homoscedastic case, $\delta=0$, are displayed in Figure \ref{fig8}. From Figure \ref{fig5}, it can be seen that the normal approximation is not met for small samples. As the sample is increased, the results become better. Figure \ref{fig6} and \ref{fig8} show that the estimator $\widehat i_{r,q}$ seems to converge to the true value, whether $i_{r,q}>0$ (the heteroscedastic case) or when $i_{r,q}=1$ is in the border of the possible values (the homoscedastic case). All these observations illustrate the statements of Proposition \ref{asy_index}. Finally, the ability of the index to detect heteroscedastic situations, is studied by allowing values from zero to one for $\delta$. From Figure \ref{fig7}, it emerges that the $\widehat i_{r,q}$ clearly take increasing values as $\delta$ is far from zero. This suggests that the proposed index is relevant to decide whether the approximated or averaged OIRF should be applied.

\subsection{Real data analysis}
\label{realdata}

We assess the discrepancy between the approximated and the averaged OIRF, for the $\log$ first differences of the brent crude in USD per barrel multiplied by 100, and the growth rate previous period for the consumer price index for the United States. The series taken from October, 2001 to June, 2020 ($T=225$) are plotted in Figure \ref{data}.\footnote{The data can be downloaded from the website of the research division of the Federal Reserve Bank of Saint Louis https: $//$fred.stlouisfed.org$/$} The effects of energy prices shocks on other macroeconomic variables are commonly investigated in the applied econometric literature. This can be explained by the importance of the energy sector in world economies or finance markets. The reader is referred to papers published in specialized journals like \textit{Energy Economics}, \textit{Energy Policy} or papers with JEL codes \textit{Q43: energy and macroeconomics} and \textit{C32: time series models}.
In general, such kind of data may exhibit fast variance changes. At first glance, this suggests that our methodology can deliver a quite different picture of the OIRF when compared to the standard approach.

First a VAR(1) model is adjusted to the series to capture the conditional mean. Following the ordering argument of L\"{u}tkepohl (2005,p61), the first component corresponds to the $\log$ differences of the oil prices and the second one to the inflation data. Indeed, it is reasonable to think that there is no instantaneous effects from the inflation to the oil prices. The model adequacy is checked using the portmanteau tests proposed in Patilea and Ra\"{\i}ssi (2013). The existence of second order dynamics in the residuals is tested using the tools proposed in Patilea and Ra\"{\i}ssi (2014). Our outputs, not displayed here, show that a deterministic specification for the variance structure seems adequate.

Now we turn to the analysis of the time varying OIRF. More precisely, we aim to compare the pre and post crisis response of the US inflation to a shock of the oil price. The pre-crisis period goes from  October, 2001 to July 2008, and the post crisis from June 2009 to January 2020. In Figure \ref{fig9}, the time-varying OIRF in the case of heteroscedastic VAR with constant conditional mean parameters are displayed. It emerges that the OIRF are subject to constant changes, including for the pre and post crisis periods. In addition, it is found that the index for the pre and post crisis periods, given in Table \ref{tab5}, is somewhat far from one. All these observations suggest to consider the averaged OIRF, in addition to the usual approximated OIRF. From Figure \ref{fig10}, it can be seen that the approximated OIRF leads to an over-estimation of the impact of oil price changes on the inflation in the United States. In particular, we found that the approximated OIRF is up to 10\% larger than the averaged OIRF. As a conclusion, our real data analysis shows that the standard approach, which consists in computing the approximate OIRF, can be quite misleading in presence of heteroscedasticity. Indeed, considering the approximated OIRF leads to an oversized estimation of the OIRF in general. This would occur especially when economic crises, or specific political events, generate smooth fast or abrupt changes in the variance of the variables. Noting that by periods analyses are actually performed to compare pre and post situations related to such events, it clearly appears that the averaged OIRF provide a reliable estimation.

\newpage

\section{Appendix}

\color{black}

\subsection{Kernel estimates of the covariance function integrals}\label{NP_integrals}


As mentioned in section \ref{adapti}, we need suitable estimators of integral functionals 
$$
\frac{1}{q}\int_{r-q/2}^{r+q/2} A(v)   \Sigma (v)  dv
$$
with $A(\cdot)$ some given matrix-valued function.
The estimator of such integrals obtained by plugging in the nonparametric estimator of the covariance structure introduced in equation (\ref{Sig_NP}) would be asymptotically biased due to boundary effects.

To explain the rationale of the alternative nonparametric estimator we propose, we will assume for the moment that the $d\times d-$matrices $u_j u_j^\prime$ are available for all $1\leq j \leq T$.
Let us consider the generic real-valued random quantity
$$
  S_T(r) \!= 
\! \int_{r-q/2}^{r+q/2} \!\!a(v) \!\left[\! \frac{1}{|\mathcal{J}| }
  \sum_{j\in \mathcal{J}}
  \omega_{ v,j}(h)   u^{(k)}_j u^{(l)}_j\! \right]  \!dv\!
=\frac{1}{|\mathcal{J}| }
\sum_{j\in \mathcal{J}}
\!\left[ 
\int_{r-q/2}^{r+q/2} \!\!a(v) \omega_{ v,j}(h)      dv \right] \!\!u^{(k)}_j u^{(l)}_j\!,
$$
where
$\mathcal{J}=\{j_{min},\ldots,j_{max}\}\subset \{1,\ldots,T\}$ is a set of consecutive indices that will be specified below and $|\mathcal{J}|= j_{max}-j_{min}+1$ is the cardinal of $\mathcal{J}$;
$\omega_{ v,j}(h) = h^{-1}L(h^{-1} (v - j/T))$ with $h$  a deterministic bandwidth with a rate that will be specified below, and $L(\cdot)$ is a bounded  symmetric density function with support $[-1,1]$;
$a ( \cdot )$ is a  given differentiable  function with Lipschitz continuous derivative;
$u^{(k)}_j$ and  $u^{(l)}_j$ are components of $u_j$ and $E(u^{(k)}_j u^{(l)}_j)=\Sigma^{(k,l)}(j/T)$, that is the $(k,l)$ cell of the matrix $\Sigma(j/T)$.


By a change of variables and Taylor expansion,
\begin{multline*}
  \int_{r-q/2}^{r+q/2} a(v) \frac{1}{h} L\left( \frac{v-j/T}{h} \right) dv
=\int_{(r-q/2-j/T)/h}^{(r+q/2-j/T)/h} a(j/T + uh) L\left( u \right) du
\\ = a(j/T) 
\int_{(r-q/2-j/T)/h}^{(r+q/2-j/T)/h}  L\left( u \right) du+ ha^\prime (j/T) 
\int_{(r-q/2-j/T)/h}^{(r+q/2-j/T)/h}  uL\left( u \right) du + O(h^2). \color{black}
\end{multline*}
To avoid large bias, we aim at using the properties $\int_{-1}^1 L(u)du = 1$ and $\int_{-1}^1 uL(u)du = 0$. For this purpose, any $j\in\mathcal{J}$ should satisfy the conditions
$
(r+q/2-j/T)/h \geq 1$  and $(r-q/2-j/T)/h\leq -1.
$
That is, the indices set $ \mathcal{J}$ should be defined such that
$$
\forall j\in\mathcal{J}, \quad (r-q/2+h)T \leq j\leq (r+q/2-h)T.
$$
Let us define
$$
j_{min} = [(r-q/2+h)T]+1\qquad \text{and} \qquad j_{max} = [(r+q/2-h)T].
$$
  Then, uniformly with respect to $j\in\mathcal{J}$,
$$
\int_{r-q/2}^{r+q/2} a(v) \frac{1}{h} L\left( \frac{v-j/T}{h} \right) dv
- a(j/T)  = O(h^2).
$$
Note that  $|\mathcal{J}| = [(q - 2h)T]$ and $|\mathcal{J}|/(q-2h)T= 1 + O(1/T)$.
\color{black}

Now, we could deduce
\begin{multline*}
  S_T =\frac{1}{|\mathcal{J}| } \sum_{j\in \mathcal{J}} a(j/T) \{u^{(k)}_j u^{(l)}_j - \Sigma^{(k,l)}(j/T)\}
\\+ \frac{1}{ |\mathcal{J}|} \sum_{j\in \mathcal{J}} a(j/T) \Sigma^{(k,l)}(j/T) + O(h^2)\\
=: \Delta_T +  \frac{1}{q-2h} \color{black} \int_{r-q/2+h}^{r+q/2-h} a(v)\Sigma^{(k,l)}(v) dv + O(T^{-1})+ O(h^2).
\end{multline*}
Let us comment on these findings. To make the reminder $O(h^2)$ negligible, we will need to  impose
$
Th^4 \rightarrow 0.
$
For instance, we could consider a bandwidth $h$ under the form
$$
h = c \frac{q}{2\sqrt{3}}T^{-2/7},\qquad \text{for some constant } c>0.
$$
The factor $q/2\sqrt{3}$ takes into account the standard deviation of a uniform design  on the interval $[r-q/2,r+q/2]$.
The term $\Delta_T$ is a sum of independent centered variables and will have a Gaussian limit. Finally, let us focus on the last integral and notice that
$$
  \frac{1}{q-2h}   \int_{r-q/2+h}^{r+q/2-h} a(v)\Sigma^{(k,l)}(v) dv = \frac{1}{q} \int_{r-q/2}^{r+q/2} a(v)\Sigma^{(k,l)}(v) dv + O(h).
$$
Thus $S_T$ preserves a non negligible bias as an estimator of $q^{-1} \int_{r-q/2}^{r+q/2} a(v)\Sigma^{(k,l)}(v) dv$. The solution we will propose to remove this bias is to define estimates like $S_T$ with modified   $q$ and thus with modified  bounds $j_{min} $ and $j_{max} $ of the set $\mathcal J$. \color{black}

\subsection{Assumptions}
\label{sec.assump}


\textbf{Assumption A0:}\: (a) The process $(\epsilon_t)$ is iid such that $E(\epsilon_t\epsilon_t')=I_d$, with $I_d$ the $d\times d$ identity matrix, \color{black} and $\sup_t\parallel\epsilon_{i,t}\parallel_{\mu}<\infty$
for some $\mu>8$ and
for all $i\in\{1,\dots,d\}$ with $\parallel .\parallel_\mu:=(E\parallel .\parallel^\mu)^{1/\mu}$ \color{black}
and $\parallel .\parallel$ being the Euclidean norm. Moreover $E\left(\epsilon_{t}^{(i)}\epsilon_{t}^{(j)}\epsilon_{t}^{(k)}\right)=0$, $i,j,k\in\{1,\dots,d\}$.

(b) 
{The matrix $\mathbb{A}$ given in (\ref{calA}) is of full rank.}

\vspace{0.3 cm}

The covariance of the system (\ref{hetero1}) is allowed to vary
in time according to assumption {\bf A1} below.

\vspace{0.3 cm}

\textbf{Assumption A1:}\: We assume that $H_{t}=G(t/T)$, where the matrices $G(\cdot)$ are lower triangular matrices with positive diagonal components. The components $\{g_{k,l}(r): 1\leq k,l \leq d\}$ of the matrices $G(r)$
are measurable deterministic functions on the interval $(0,1]$, with
$\forall \,1\leq k,l\leq d$, $ \sup_{r\in(0,1]}|g_{k,l}(r)|<\infty$. The functions $g_{k,l}(\cdot)$ satisfy a Lipschitz condition piecewise on a finite partition of $(0,1]$ in sub-intervals (the partition may depend on $k,l$).
The matrix $\Sigma(r)=G(r)G(r)'$ is assumed positive definite for all $r$ \color{black} and $\inf_{r\in(0,1]} \lambda_{min}(\Sigma(r)) >0$ where $ \lambda_{min}(\Gamma)$ denotes the smallest eigenvalue of the symmetric matrix $\Gamma$.

\vspace{0.3 cm}

{The OIRF estimates we investigate in the following are obtained as products between a functional of the innovation vectors $u_t$ (the estimator of $\varphi_0$) and a centered functional of matrices $u_tu_t^\prime$ (the estimator of some square root matrix built using the covariance structure $\Sigma(\cdot)$). The $\sqrt{T}-$asymptotic normality of the OIRF estimators is then deduced from the asymptotic behavior of the two factors. The condition $E\left(\epsilon_{t}^{(i)}\epsilon_{t}^{(j)}\epsilon_{t}^{(k)}\right)=0$, $i,j,k\in\{1,\dots,d\}$, is a convenient condition for simplifying the asymptotic variance of our estimators, that is making it block diagonal. It is in particular fulfilled if the errors are supposed Gaussian. The asymptotic results could be also deduced if this condition fails, the asymptotic variance of the estimators would then include some additional covariance terms.
}



\vspace{0.3 cm}

\textbf{Assumption A2:} \, (i) The kernel $K(\cdot)$ is a bounded \color{black} symmetric \color{black} density function defined on the real line such that $K(\cdot)$ is nondecreasing on $(-\infty, 0]$ and decreasing on $[0,\infty)$ and $\int_\mathbb{R} |v|K(v)dv < \infty$. The function $K(\cdot)$ is differentiable except a finite number of points and the derivative $K^\prime(\cdot)$  is a bounded  integrable function.
Moreover, the Fourier Transform $\mathcal{F}[K](\cdot)$ of $K(\cdot)$ satisfies $\int_{\mathbb{R}}  \left| s \mathcal{F}[K](s) \right|ds <\infty$.

(ii) The bandwidths $b_{kl}$, $1\leq k\leq l\leq d$, are taken in the range $\mathcal{B}_T = [c_{min} b_T, c_{max} b_T]$ with $0< c_{min}< c_{max}< \infty$ and $b_T + 1/Tb_T^{2+\gamma} \rightarrow 0$ as $T\rightarrow \infty$, for some $\gamma >0$.

\color{black}
(iii) The kernel $L(\cdot)$ is a symmetric bounded Lipschitz continuous density function with support in $[-1,1]$.

(iv) The bandwidth $h$ satisfies the condition $ h^4 T+ 1/Th^2  \rightarrow 0$ as $T\rightarrow \infty$.

\color{black}

\vspace{0.3 cm}

\subsection{ Proofs}

In the sequel, $c$, $c^\prime$, $c^{\prime\prime}$ and $C$, $C^\prime$, $C^{\prime\prime}$ are constants, possibly different from line to line.

\vspace{0.3cm}

\begin{proof}[Proof of Lemma \ref{CS_ineq_m}]
First, note that
\begin{multline*}
\widetilde H(r) \widetilde H(r)^\prime - \bar H(r)\bar H(r)^\prime =  \frac{1}{q}\int_{r-q/2}^{r+q/2}H(v)H(v)^\prime dv \\ - \left[ \frac{1}{q}\int_{r-q/2}^{r+q/2}\!H(v)dv \right]\! \left[ \frac{1}{q}\int_{r-q/2}^{r+q/2}\!H(v)dv \right]^\prime
\end{multline*}
is a positive semi-definite matrix, whatever the values of $r$ and $q$ are. Moreover,
\begin{equation}\label{statem}
\widetilde H(r) \widetilde H(r)^\prime = \bar H(r)\bar H(r)^\prime \text{ if and only if $H(\cdot)$ is constant on $(r-q/2, r+q/2)$}
\end{equation}
Indeed, for any $a\in\mathbb{R}^d$,
\begin{multline*}
a^\prime \left\{ \widetilde H(r) \widetilde H(r)^\prime - \bar H(r)\bar H(r)^\prime\right\} a  \\
=\frac{1}{q} \int_{r-q/2}^{r+q/2}a^\prime \left[H(v)-  \frac{1}{q}\int_{r-q/2}^{r+q/2}H(u)du \right] \left[H(v)-  \frac{1}{q}\int_{r-q/2}^{r+q/2}H(u)du \right]^\prime a \;dv\\
= \frac{1}{q} \int_{r-q/2}^{r+q/2}\left\|a^\prime \left[H(v)-  \frac{1}{q}\int_{r-q/2}^{r+q/2}H(u)du \right] \right\|^2 dv \geq 0.
\end{multline*}
This shows that $\widetilde H(r) \widetilde H(r)^\prime - \bar H(r)\bar H(r)^\prime$ is positive semi-definite.
Next, under our assumptions, for each $a\in\mathbb{R}^d$ the map
\begin{equation}\label{erz}
v\mapsto \left\|a^\prime \left[H(v)-  \frac{1}{q}\int_{r-q/2}^{r+q/2}H(u)du \right] \right\|^2
\end{equation}
is piecewise continuous on $(0,1)$. Thus, if
$\widetilde H(r) \widetilde H(r)^\prime = \bar H(r)\bar H(r)^\prime$, then  necessarily, for each $a$, the map (\ref{erz}) is constant equal to zero. This implies that $H(\cdot)$ is constant on $(r-q/2, r+q/2)$. Conversely, when $H(\cdot)$ is constant, then $\widetilde H(\cdot) = \bar H(\cdot)$ and thus $\widetilde H(r) \widetilde H(r)^\prime = \bar H(r)\bar H(r)^\prime$. Finally, the two statements in the lemma are direct consequences of \eqref{statem} and the positive semi-definiteness of
$\widetilde H(r) \widetilde H(r)^\prime - \bar H(r)\bar H(r)^\prime$.
\end{proof}

\vspace{0.3cm}

\begin{proof}[Proof of Proposition \ref{prop_tv_oirf}.]
\color{black}
For the convergence of $\widehat{\Sigma}_{[rT]}$ let us recall that $\Sigma(r)=G(r)G(r)'$ and the components $\{g_{k,l}(\cdot): 1\leq k,l \leq d\}$ of $G(\cdot)$ are bounded piecewise Lipschitz continuous functions. Let  $U_t(\vartheta)=u_t(\vartheta)u_t(\vartheta)'$ with $u_t(\vartheta)=X_t-(\widetilde{X}_{t-1}'\otimes I_d)\vartheta$ for some $\vartheta\in R^{d^2p}$. Thus $U_t(\vartheta_0)=u_tu_t'$ and  $U_t(\widehat \vartheta_{OLS})=\widehat u_t \widehat u_t'$. By elementary matrix algebra, 
\begin{multline*}
\left\| U_t(\widehat \vartheta_{OLS}) - U_t(\vartheta_0) \right\|_F \leq 2 d \sqrt{p} \|G\|_{\infty}\left\| \widehat \vartheta_{OLS} - \vartheta_0 \right\| \left\| \widetilde{X}_{t-1} \right\| \|\epsilon_t\|\\
+ d^2p  \left\| \widehat \vartheta_{OLS} - \vartheta_0 \right\|^2\left\| \widetilde{X}_{t-1} \right\| ^2.
\end{multline*}
Herein, $\|\cdot\|_F$, $\|\cdot\|$ and $\|\cdot\|_\infty$ are the Frobenius, Euclidian and uniform norms, respectively.  By the triangle inequality,  the monotonicity of $K(\cdot)$ and the rate of $\left\| \widehat \vartheta_{OLS} - \vartheta_0 \right\|$, deduce
\begin{multline*}
\sup_{\mathcal{B}_T} \! \left\| \widehat{\Sigma}_{[rT]} - \! \sum_{j=1}^{T} w_{[rT],j} \odot {u}_j{u}_j^\prime \right\|_F \! \! \leq \frac{1}{c_{min}Tb_T} \sum_{j=1}^{T} K\left( \! \frac{[rT]-j}{c_{max}b_T T} \right) \! \left\|   \widehat{u}_j\widehat{u}_j^\prime - u_j u_j'\right\|_F \\ = O_{\mathbb{P}} (1/\sqrt{T}).
\end{multline*}
Next, we can write
\begin{multline*}
 \sum_{j=1}^{T} w_{[rT],j} \odot {u}_j{u}_j^\prime=  \sum_{j=1}^{T}  w_{[rT],j}   \odot \left\{ {u}_j{u}_j^\prime  -  E\left( {u}_j{u}_j^\prime \right)\right\}+ \sum_{j=1}^{T}   w_{[rT],j}   \odot  E\! \left( {u}_j{u}_j^\prime \right) \\  =: \Sigma_{1,[rT]}+\Sigma_{2,[rT]}.
\end{multline*}
Let $\sigma_{1,[rT]}^{(k,l)}$, $\sigma_{2,[rT]}^{(k,l)}$, $\Sigma(r)^{(k,l)}$and $\Sigma_j^{(k,l)}$ denote the $(k,l)$ elements of the matrices $\Sigma_{1,[rT]}$, $\Sigma_{2,[rT]}$, $\Sigma(r)$ and $E ( {u}_j{u}_j^\prime )$, respectively.

First we study the bias. For any $r\in (0,1)$, since $K(\cdot)$ is symmetric, we have
\begin{multline*}
\sigma_{2,[rT]}^{(k,l)}= \frac{1}{Tb_{kl}}\sum_{j=1}^T K\left(  \frac{j-[rT]}{Tb_{kl}} \right) \Sigma_j^{(k,l)} \\ = \frac{1}{b_{kl}} \int_{[1/T, (1+T)/T)} K\left(  \frac{[sT]-[rT]}{Tb_{kl}} \right) \Sigma_{[sT]}^{(k,l)} ds 
\\\stackrel{z=(s-r)/b_{kl}}{=}   \int_{[(1-Tr)/Tb_{kl}, (1+T-Tr)/Tb_{kl})} K\left(  \frac{[(r+zb_{kl})T] - [rT]}{Tb_{kl}} \right) \Sigma(r+zb_{kl})^{(k,l)} dz \\
=   \int_{[-r/b_{kl}, 0)} K\left(  \frac{[(r+zb_{kl})T]-[rT]}{Tb_{kl}} \right) \left\{\Sigma(r+zb_{kl})^{(k,l)} - \Sigma(r+)^{(k,l)} \right\} dz \\
+  \int_{[0 , (1-r)/b_{kl})} K\left(  \frac{[(r+zb_{kl})T]-[rT]}{Tb_{kl}} \right) \left\{\Sigma(r+zb_{kl})^{(k,l)} - \Sigma(r-)^{(k,l)} \right\} dz \\
+ \Sigma(r+)^{(k,l)} \int_{[-r/b_{kl}, 0)} K\left(  \frac{[(r+zb_{kl})T]-[rT]}{Tb_{kl}} \right)dz \\ + \Sigma(r-)^{(k,l)} \int_{[0 , (1-r)/b_{kl})} K\left(  \frac{[(r+zb_{kl})T]-[rT]}{Tb_{kl}} \right)   dz + O(1/Tb_T) .
\end{multline*}
Next, on the intervals  where the Lipschitz property holds true, for $z\geq 0$ we have
$$
\left| \Sigma(r+zb_{kl})^{(k,l)} - \Sigma(r+)^{(k,l)} \right| \leq L c_{max} |z|b_T,
$$
and for $z<0$ we have
$$
 \left| \Sigma(r+zb_{kl})^{(k,l)} - \Sigma(r-)^{(k,l)} \right| \leq L c_{max} |z| b_T,
$$
for some constant $L$. Meanwhile, for any $b_{kl}\in\mathcal{B}_T$,
$$
 0 \leq \frac{[(r+zb_{kl})T] - [rT]}{Tb_{kl}} - z  \leq \frac{1}{T c_{min} b_T },
$$
so that, since $K(\cdot)$ is piecewise Lipschitz continuous, except at most a finite number of values $z$,
$$
\sup_{b_{kl} \in\mathcal{B}_T}\left|K\left(  \frac{[(r+zb_{kl})T]-[rT]}{Tb_{kl}} \right) - K(z)  \right| \leq \frac{1}{T c_{min} b_T }.
$$
Finally, for $r \in (0,1)$, $r/b_{kl}$ and $(1-r)/b_{kl}$ tend to infinity and thus
$$
\inf_{b_{kl} \in\mathcal{B}_T} \int_{[-r/b_{kl}, 0)} K \left( \frac{[(r+zb_{kl})T]-[rT]}{Tb_{kl}} \right) dz \uparrow \frac{1}{2}
$$
and
$$
\inf_{b_{kl} \in\mathcal{B}_T}  \int_{[0 , (1-r)/b_{kl})}  K \left( \frac{[(r+zb_{kl})T]-[rT]}{Tb_{kl}} \right) \uparrow \frac{1}{2}.
$$
The case $r=1$ could be treated with similar arguments. Gathering facts, deduce that, for any $r\in(0,1]$, the rate of the bias term is
$$
\sup_{\mathcal{B}_T} \left\| \widehat{\Sigma}_{2,[Tr]} - \frac{1}{2} \left\{\Sigma(r-) + \Sigma(r+)\right\}\right\|_F  = O\left( b_T+1/Tb_T \right).
$$

For the variance term $\Sigma_{1,[rT]}$, we could use the properties of the empirical process indexed by families of functions of polynomial complexity. Here the family of functions are indexed by the constants that multiplies the rate $b_T$ to define the bandwidths for each element $(k,l)$ in the matrix. The polynomial complexity is guaranteed by the monotonicity of $K(\cdot)$ and by the fact that the polynomial complexity is preserved by finite unions.  We apply  Theorem 3.1 in van der Vaart and Wellner (2011) for each component $(k,l)$ with the family
$$
\mathcal{F}_T = \{ (r,u^{(k)},u^{(l)}) \mapsto K(a - r/cb_T)  u^{(k)}u^{(l)}: r\in(0,1], a>0,c_{min}\leq c \leq c_{max}  \},
$$
the envelope $F(r,u^{(k)},u^{(l)})  = Cu^{(k)}u^{(l)}$ for some constant $C>0$,   $p=(\mu-4)/(\mu - 8)>1$, $\delta^2=c^\prime b_T$ for some constant $c^\prime >0$. In this case $J(\delta,\mathcal F, L_2) \leq C^\prime  \delta\sqrt{\log(1/\delta)} \leq C^{\prime \prime } \sqrt{b_T\log (T)}$, for some constants $C^\prime, C^{\prime\prime}>0$.
Deduce that
$$
\sup_{\mathcal{B}_T} \left\| \widehat{\Sigma}_{1,[Tr]} \right\|_F  = O_{\mathbb{P}} \left( \sqrt{\log(T)/Tb_T} \right).
$$
For the second part of the results on the matrices $H$, it suffices to apply a perturbation bound for the Cholesky factorization, as for instance in Theorem 3.1 of Chang and Stehl{\'e} (2010), to deduce
$$
\sup_{\mathcal{B}_T} \left\| \widehat{H}_{[Tr]} -  \frac{1}{2}  H_{\pm}(r)  \right\|_F \leq C \left\|  \frac{1}{2}  H_{\pm}(r)  \right\|_F \frac{\sup_{\mathcal{B}_T} \left\| \widehat{\Sigma}_{[Tr]} - \frac{1}{2} \left\{\Sigma(r-) + \Sigma(r+)\right\}\right\|_F  }{\left\|  \frac{1}{2} \left\{\Sigma(r-) + \Sigma(r+)\right\}\right\|_F },
$$
for some constant $C$. Now the proof in complete.
\end{proof}

\vspace{0.3cm}

\begin{proof}[Proof of Lemma \ref{lemma_aoirf}.]
First let us write from the mean value Theorem:
\begin{equation*}
\mbox{vech}(Int_{T,r,q} (\widehat{U}) )=\mbox{vech}(Int_{T,r,q} (U) )+\frac{\partial\mbox{vech}(Int_{T,r,q} (U(\vartheta)) )}{\partial\vartheta'}|_{\vartheta=\vartheta^*}(\widehat{\vartheta}_{OLS}-\vartheta_0),
\end{equation*}
where $\widehat{U}_t=\hat{u}_t\hat{u}_t'$, $U_t=u_tu_t'$, $U_t(\vartheta)=u_t(\vartheta)u_t(\vartheta)'$ and $u_t(\vartheta)=X_t-(\widetilde{X}_{t-1}'\otimes I_d)\vartheta$ for some $\vartheta\in R^{d^2p}$ and $\vartheta^*$ between $\widehat{\vartheta}_{OLS}$ and $\vartheta_0$. Noting that $\frac{\partial u_t(\vartheta)}{\partial\vartheta'}=-(\widetilde{X}_{t-1}'\otimes I_d)$, the consistency of $\widehat{\vartheta}_{OLS}$ and the fact that $u_t$ is not correlated with $\widetilde{X}_t$, we obtain using basic derivative rules:

$$\frac{\partial\mbox{vech}(Int_{T,r,q} (U(\vartheta))}{\partial\vartheta'}|_{\vartheta=\vartheta^*}=o_p(1).$$
Using the $\sqrt{T}$-convergence of $\widehat{\vartheta}_{OLS}$ (see (\ref{res2ols})), this implies that
\begin{equation}\label{withouthatsame}
\sqrt{T}\mbox{vech}(Int_{T,r,q} (\widehat{U})-Int_{r,q}(\Sigma))=\sqrt{T}\mbox{vech}(Int_{T,r,q} (U)-Int_{r,q}(\Sigma))+o_p(1),
\end{equation}
where we recall that $Int_{r,q}(\Sigma)=q^{-1}\int_{r-q/2}^{r+q/2}\Sigma(v)dv$.

Next, we investigate the joint distribution of
\begin{equation}\label{joint}
\sqrt{T}\left[(\widehat{\vartheta}_{OLS}-\vartheta_0)',\left\{\mbox{vech}(Int_{T,r,q} (U) -Int_{r,q}(\Sigma))\right\}'\right]'.
\end{equation}
We write:
{\small
$$\left(\!\!
    \begin{array}{c}
      \widehat{\vartheta}_{OLS}-\vartheta_0 \\
      \mbox{vech}(Int_{T,r,q} (U)-Int_{r,q}(\Sigma)) \\
    \end{array}
  \!\!\right)
=\left(\!\!
    \begin{array}{cc}
      \left\{Int_{T,0.5,1} (\widetilde{\mathbf{X}}) )\otimes I_d\right\}^{-1} & 0 \\
      0 & I_{d(d+1)/2} \\
    \end{array}
  \!\!\right)\! \left(\!\!
           \begin{array}{c}
            \Upsilon_t^1  \\
            \Upsilon_t^2  \\
           \end{array}
        \!\! \right),
$$}
where $\widetilde{\mathbf{X}}_{t-1}=\widetilde{X}_{t-1}\widetilde{X}_{t-1}'$, $\Upsilon_t^1=\mbox{vec}(Int_{T,0.5,1} (X^u) )$, with $X^u_{t}=u_t\widetilde{X}_t'$ and
$$\Upsilon_t^2=\mbox{vech}(Int_{T,r,q} (U) -Int_{r,q}(\Sigma)).$$ The vector
$\Upsilon_t=(\Upsilon_t^{1'},\Upsilon_t^{2'})'$ is a martingale difference since the process $(u_t)$ is independent. On the other hand we have
$T^{-1}\sum_{t=1}^{T}Int_{T,0.5,1} (\widetilde{\mathbf{X}})\otimes I_d\rightarrow\Lambda_3$, from Patilea and Ra\"{\i}ssi (2012). Then from the Lindeberg CLT and the Slutsky Lemma, (\ref{joint}) is asymptotically normally distributed with mean zero. For the asymptotic covariance matrix in (\ref{res0}), the top left block is given from the asymptotic normality result (\ref{res2ols}), while the bottom right block can be obtained using the same arguments of Patilea and Ra\"{\i}ssi (2010), Lemma 7.1, 7.2, 7.3 and 7.4. The asymptotic covariance matrix is block diagonal since we assumed that $E(u_{it}u_{jt}u_{kt})=0$, $i,j,k\in\{1,\dots,d\}$ in {\bf A1}, together with considering again that $u_t$ is independent with respect to the past of $X_t$. Hence the asymptotic matrix of (\ref{joint}) is given as in (\ref{res0}).
\end{proof}

\vspace{0.3cm}

\color{black}
To simplify the reading, before proceeding to the next proofs, let us put the orthogonal impulse response function (OIRF) notation in a nutshell. First,  let $S\mapsto C(S)$ be the operator that maps a positive definite matrix into the lower triangular matrix of the Cholesky decomposition of $S$. Next, consider a matrix-valued function $r\mapsto A(r)$,  $r\in (0,1]$, and, for any $r\in (0,1]$, $0<q<1$ such that $0 <r-q/2 < r +q/2 <1$, let
$$
Int_{r,q} (A) = \frac{1}{q}\int_{r-q/2}^{r+q/2} A(v)dv \qquad  Int_{T,r,q} (A) =   \frac{1}{[qT]+1}\sum_{k=-[qT/2]}^{[qT/2]}  A_{[rT]-k}\color{black}.
$$
If $\sup_{r\in (0,1)} \|A(r)\|_F<\infty$ and the components of $A(\cdot)$
are piecewise Lipschitz continuous on each sub-intervals of a finite number  partition of $(0,1]$, then
there exists a constant $c$ such that
$$
\sup_{r,q} \left\|Int_{r,q} (A) - Int_{T,r,q} (A) \right\|_F \leq c T^{-1}.
$$
Now,  we could rewrite the theoretical IRF we introduced above as follows: for any $i\geq 1$,
$$
\text{(approximated OIRF)}  \qquad \widetilde{\theta}_{r}^{q}(i)=\Phi_i\widetilde{H}(r) = \Phi_i C(Int_{r,q}(\Sigma)),
$$
and
$$
\text{(averaged OIRF)} \qquad \bar{\theta}_{r}^{q}(i):=\Phi_i\left\{\frac{1}{q}\int_{r-q/2}^{r+q/2}H(v)dv\right\} =  \Phi_i Int_{r,q}(C(\Sigma)).
$$
Moreover, the estimators we introduced could be rewritten as follows: with the matrix-valued function $r\mapsto \widehat{U}(r)=\widehat{u}_{[rT]}\widehat{u}_{[rT]}'$, the usual approximated OIRF estimator
is
$$
\hat{\widetilde{\theta}}_{r}^{q}(i)=\widehat{\Phi}_i^{ols}\widehat{\widetilde H}(r) 
= \widehat{\Phi}_i^{ols} C\left(Int_{T,r,q} \left(\widehat U\right) \right);
$$
the  new approximated OIRF estimator 
is
$$
\hat{\widetilde{\theta}}_{r}^{q,als}(i)=
\widehat{\Phi}_i^{als} C\left(Int_{T,r,q} \left(\widehat U\right)\right);
$$
and the averaged OIRF estimator 
is
$$
\hat{\bar{\theta}}_{r}^{q}(i)= 
\widehat{\Phi}_i^{als} \frac{q+h}{q}  Int_{T,r, q+h} \left(C\left(\widehat V\right)\right),
$$
with
$$
\frac{q+h}{q} Int_{T,r,q+ h} \left(C\left(\widehat V\right)\right) = \bar{\widehat{H}}(r)  \{1+ O(1/T)\}
= 
 \frac{1+ O(1/T)}{[qT]+1}\sum_{k=-[( q+h)T/2]}^{[( q+h)T/2]}\widehat{H}_{[rT]-k} ;
$$
and  $\widehat H_{[rT]-k }$  the lower triangular matrix of the
Cholesky decomposition of ${\widehat V}_{[rT]-k} $, where $\widehat V_{[rT]-k} $ is defined in equation (\ref{nnapprx2}).


Next, let us recall the differentiation formula of the Cholesky operator
\begin{equation}\label{diff_Ch}
\Delta:= \frac{\partial\mbox{vec} (C (\Sigma))}{\partial\mbox{vec} (\Sigma)} = (I_d\otimes  C (\Sigma)) Z (C (\Sigma)^{-1}\otimes C (\Sigma)^{-1}) ,
\end{equation}
where $Z$ is a diagonal matrix such that $Z  \mbox{vec} (A) = \mbox{vec} (\Phi (A))$ for any $d\times d-$matrix $A$. Here $\Phi$  takes the lower-triangular part
of a matrix and halves its diagonal:
$$
\Phi(A)_{ij} = \left\{
\begin{tabular}{ll}
$A_{ij}$ & \; $i>j$ \\
$\frac{1}{2}A_{ij}$  & \; $i=j$\\
0 & \; $i<j$
\end{tabular}
\right. .
$$
Note that
$$
 (C (\Sigma)^{-1}\otimes C (\Sigma)^{-1}) \mbox{vec} \left( \Sigma \right) = \mbox{vec}\left(C (\Sigma)^{-1}  \Sigma(C (\Sigma)^\prime)^{-1} \right) =  \mbox{vec}(I_d)
$$
and thus
\begin{multline*}
\Delta\mbox{vec} \left(\Sigma \right) = (I_d\otimes  C (\Sigma)) Z\mbox{vec}(I_d)  =(I_d\otimes  C (\Sigma)) \mbox{vec} (\Phi (I_d)) \\ = \mbox{vec}(C (\Sigma) \Phi (I_d)) = \frac{1}{2}  \mbox{vec}(C (\Sigma) ).
\end{multline*}

\vspace{0.3cm}

\begin{proof}[Proof of Proposition \ref{resapprox}.]
Using our notations we write
$$
\hat{\widetilde{\theta}}_{r}^{q}(i)-\widetilde{\theta}_{r}^{q}(i)=\widehat{\Phi}_i^{ols}C(Int_{T,r,q}(\widehat{U}))-\Phi_iC(Int_{r,q}(\Sigma)).
$$
From (\ref{withouthatsame}) and the consistency of the OLS estimator, we have
\begin{multline*}
\!\!\!\sqrt{T}(\widehat{\Phi}_i^{ols}C(Int_{T,r,q}(\widehat{U}))\!-\!\Phi_iC(Int_{r,q}(\Sigma)))\!=\!
\sqrt{T}(\widehat{\Phi}_i^{ols}C(Int_{T,r,q}(U))\!-\!\Phi_iC(Int_{r,q}(\Sigma)))\\+o_p(1).
\end{multline*}
Now let us write
\begin{multline}
  \sqrt{T}\mbox{vec}\left[\widehat{\Phi}_i^{ols}C(Int_{T,r,q}(U)-\Phi_iC(Int_{r,q}(\Sigma))\right] = \sqrt{T}\mbox{vec}\left[(\widehat{\Phi}_i^{ols}-\Phi_i)C(Int_{r,q}(\Sigma))\right.\label{rhs1} \\
  + \Phi_i(C(Int_{T,r,q}(U)-C(Int_{r,q}(\Sigma)) \\
  + \left.(\widehat{\Phi}_i^{ols}-\Phi_i)(C(Int_{T,r,q}(U))-C(Int_{r,q}(\Sigma)))\right].
\end{multline}
For the third term in the right hand side of (\ref{rhs1}), $\sqrt{T}\mbox{vec}\{\widehat{\Phi}_i^{ols}-\Phi_i\}$ is asymptotically normal as we can apply the delta method from {\bf A0}(b), Lemma \ref{lemma_aoirf} and Rule (8) Appendix A.13 in L\"{u}tkepohl (2005). Similarly using Rule (10) in Appendix A.13 of L\"{u}tkepohl (2005) and Lemma \ref{lemma_aoirf} again,
$\sqrt{T}\mbox{vec}\{(C(Int_{T,r,q}(U))-C(Int_{r,q}(\Sigma)))\}$ is asymptotically normal. Hence we have
$$
(\widehat{\Phi}_i^{ols}-\Phi_i)(C(Int_{T,r,q}(U))-C(Int_{r,q}(\Sigma)))=O_p(T^{-1}),
$$
so that we obtain
\begin{multline*}
  \sqrt{T}\mbox{vec}\left[\widehat{\Phi}_i^{ols}C(Int_{T,r,q}(U)-\Phi_iC(Int_{r,q}(\Sigma))\right]= \sqrt{T}\mbox{vec}\left[(\widehat{\Phi}_i^{ols}-\Phi_i)C(Int_{r,q}(\Sigma))\right. \\
  + \Phi_i(C(Int_{T,r,q}(U)-C(Int_{r,q}(\Sigma))\Big]+o_p(1).
\end{multline*}
For the right-hand side of this equation, by the identity vec$(ABC)=(C'\otimes A)\mbox{vec}(B)$ for matrices of adequate dimensions:
$$
\mbox{vec}\left\{(\widehat{\Phi}_i^{ols}-\Phi_i)C(Int_{r,q}(\Sigma))\right\}=(C(Int_{r,q}(\Sigma))'\otimes I_d)\mbox{vec}(\widehat{\Phi}_i^{ols}-\Phi_i),
$$
and
\begin{multline*}
\mbox{vec}\left\{\Phi_i(C(Int
_{T,r,q}(U)-C(Int_{r,q}(\Sigma))\right\}=(I_d\otimes\Phi_i)\mbox{vec}\{C(Int_{T,r,q}(U)-C(Int_{r,q}(\Sigma))\}\\
=  (I_d\otimes\Phi_i)\Delta \left\{ \mbox{vec}(Int_{T,r, q}(U) ) -\mbox{vec}(Int_{r,q}(\Sigma)\right\} \{1+o_{\mathbb P}(1)\}.
\end{multline*}
For the last equality we used \eqref{diff_Ch} and the delta method argument. The convergence  (\ref{histapproxasymp})
follows by Lemma \ref{lemma_aoirf}  and the CLT.
\end{proof}

\vspace{0.3cm}


\begin{proof}[Proof of Proposition \ref{propostu}]
Let us fix $\widetilde q \in (0,1)$ and consider the definitions in equations (\ref{nnapprx2}) and (\ref{accumulatedgls_H}) with the generic $\widetilde q$ replacing $q$. Note that, given $A(\cdot)$ a $d \times d-$matrix valued function defined on $(0,1]$ with differentiable elements that have Lipschitz continuous derivatives, we have
\begin{multline}\label{equiv_crucial}
Int_{T,r,\widetilde q+2h} \left(\mbox{vec} (A \widehat V) \right) = 
\frac{1}{[(\widetilde q+2h)T]+1}\sum_{k=-[(\widetilde q+2h)T/2]}^{[(\widetilde q+2h)T/2]} \mbox{vec} ({A}_{[rT]-k}{\widehat V}_{[rT]-k}) \\
=\frac{1}{\widetilde q+2h}\; \frac{1}{T}  \sum_{j=[(r-\widetilde q/2)T]+1}^{[(r+ \widetilde q/2)T]} \mbox{vec} \left(  \left[ \int_{r- \widetilde q/2-h}^{r+ \widetilde q/2+h} \frac{1}{h} L\left( \frac{v-j/T}{h} \right)  A(v)  dv \right] \widehat{u}_j \widehat{u}^{\prime}_j\right)+O_{\mathbb{P}}(1/Th) \\
=\frac{1}{\widetilde q+2h}\; \frac{1}{T}  \sum_{j=[(r- \widetilde q/2)T]+1}^{[(r+\widetilde q/2)T]} \! \mbox{vec} \left(  A(j/T) \widehat{u}_j \widehat{u}^{\prime}_j \right)+O_{\mathbb{P}}(h^2+1/Th) \\=\frac{\widetilde q}{\widetilde q+2h}\; \frac{1}{[\widetilde qT]+1 }  \sum_{j=-[\widetilde qT/2]}^{[\widetilde qT/2]}   \mbox{vec} \left(  A(([rT]-j)/T) \widehat{u}_{[rT]-j}\widehat{u}_{[rT]-j}' \right)+O_{\mathbb{P}}(h^2+1/Th)
\\= 
\frac{\widetilde q}{\widetilde q+2h}\;  Int_{T,r,\widetilde q} \left(\mbox{vec} (A\widehat U)\right) + O_{\mathbb{P}}(h^2+1/Th).
\end{multline}

Moreover,  we can write
$$
\mbox{vec} (C (\widehat V)) \! - \mbox{vec} (C (\Sigma)) = \Delta \left[ \mbox{vec} (\widehat V )\!-  \mbox{vec} \left(\Sigma\right) \right] \! \{1+o_{\mathbb{P}} (1)\}.
$$
Gathering facts, we could now study the asymptotic equivalent of $\mbox{vec}(Int_{T,r,q+h} (C(\widehat V)) )$. We have
\begin{multline*}
\mbox{vec}\left(Int_{T,r,q+h} \left(C\left(\widehat V\right)\right) \right) =Int_{T,r,q+h} \left(\mbox{vec} \left(C\left(\widehat V\right)\right) \right) \\
=Int_{r,q+h} \left(\mbox{vec} \left(C\left(\Sigma\right)\right) \right) + O_{\mathbb{P}}(1/Th)\\
+ \left\{ Int_{T,r,q+h}  (\Delta\mbox{vec}(\widehat V)) -  Int_{r,q+h}  (\Delta\mbox{vec}(\Sigma)) + O_{\mathbb{P}}(1/Th)\right\}  \{1+o_{\mathbb{P}} (1)\} \\
= \frac{q}{q+h} Int_{r,q} \left(\mbox{vec} \left(C\left(\Sigma\right)\right) \right) + O_{\mathbb{P}}(1/Th)\\+\frac{1}{q+h}\int_{r-(q+h)/2}^{r-q/2} \mbox{vec}(C\left(\Sigma(v)\right))  dv  +\frac{1}{q+h}\int^{r+(q+h)/2}_{r+q/2} \mbox{vec} (C\left(\Sigma(v)\right) )dv \\
+ \left\{\frac{q\!-\!h}{q\!+\!h}Int_{T,r,q-h}  (\Delta\mbox{vec}(\widehat U))\! + O_{\mathbb{P}}(h^2\!+\! 1/Th)\! - \! Int_{r,q+h}  (\Delta\mbox{vec}(\Sigma)) + O_{\mathbb{P}}(1/Th) \right\}\\ \times  \{1+o_{\mathbb{P}} (1)\},
\end{multline*}
where for replacing $Int_{T,r,q+h}  (\Delta\mbox{vec}(\widehat V))  $ we use the equation (\ref{equiv_crucial}) with $q-h$ instead of $\widetilde q$.
Moreover, since $\Delta\mbox{vec} \left(\Sigma \right) = (1/2)  \mbox{vec}(C (\Sigma) )$, we also have
\begin{multline*}
Int_{r,q+h}  (\Delta\mbox{vec}(\Sigma)) = \frac{q-h}{q+h} Int_{r,q-h}  (\Delta\mbox{vec}(\Sigma)) \\
+\frac{1}{2} \; \frac{1}{q+h} \int^{r-q/2}_{r-(q+h)/2} \mbox{vec} (C\left(\Sigma(v)\right) )dv+\frac{1}{2} \; \frac{1}{q+h} \int^{r-(q-h)/2}_{r-q/2}
 \mbox{vec} (C\left(\Sigma(v)\right) )dv\\
+\frac{1}{2} \; \frac{1}{q+h} \int^{r+(q+h)/2}_{r+q/2}
\mbox{vec} (C\left(\Sigma(v)\right) )dv+\frac{1}{2} \; \frac{1}{q+h} \int^{r+q/2}_{r+(q-h)/2} \mbox{vec} (C\left(\Sigma(v)\right) )dv
\\=\frac{q-h}{q+h} Int_{r,q-h}  (\Delta\mbox{vec}(\Sigma)) \\
+ \frac{1}{q+h} \int^{r-q/2}_{r-(q+h)/2} \mbox{vec} (C\left(\Sigma(v)\right) )dv+ \frac{1}{q+h} \int^{r+(q+h)/2}_{r+q/2} \mbox{vec} (C\left(\Sigma(v)\right) )dv + O(h^2),
\end{multline*}
 where for the last equality we use the change of variables $v\rightarrow v -h/2$ (resp. $v\rightarrow v + h/2$) in the integral on the interval $[r-q/2, r-(q-h)/2]$ (resp.
 $[r+q/2, r+(q+h)/2]$) and  the Lipschitz property of the elements on $\Sigma(\cdot)$. Thus, we could write
 \begin{multline}\label{id78}
 \mbox{vec}\left(Int_{T,r,q+h} \left(C\left(\widehat V\right)\right) \right) =\frac{q}{q+h} Int_{r,q} \left(\mbox{vec} \left(C\left(\Sigma\right)\right) \right) \\
 +\frac{q-h}{q+h} \left\{ Int_{T,r,q-h}  (\Delta\mbox{vec}(\widehat U)) - Int_{r,q-h}  (\Delta\mbox{vec}(\Sigma))\right\} + O_{\mathbb{P}}(h^2+1/Th).
 \end{multline}
That means
\begin{multline}\label{iid_rep_c}
\sqrt{T} \left( \frac{q+h}{q}  \mbox{vec}\left(Int_{T,r,q+h} \left(C\left(\widehat V\right)\right) \right) -\mbox{vec}  \left(Int_{r,q}\left(C\left(\Sigma\right)\right) \right) \right) \\
= \left\{\!1\! - \frac{h}{q}\right\} \sqrt{T} \left\{ Int_{T,r,q-h}  (\Delta\mbox{vec}(\widehat U)) - Int_{r,q-h}  (\Delta\mbox{vec}(\Sigma))\right\} + O_{\mathbb{P}}(\sqrt{Th^4}+1/\sqrt{Th^2})\\
=\sqrt{T} \left\{ Int_{T,r,q-h}  (\Delta\mbox{vec}(\widehat U)) - Int_{r,q-h}  (\Delta\mbox{vec}(\Sigma))\right\} + O_{\mathbb{P}}(h+\sqrt{Th^4}+1/\sqrt{Th^2}).
\end{multline}
It also means that
\begin{equation}\label{utile_q}
\frac{q+h}{q}  \mbox{vec}\left(Int_{T,r,q+h} \left(C\left(\widehat V\right)\right) \right) = \mbox{vec}  \left(Int_{r,q}\left(C\left(\Sigma\right)\right) \right)+ O_{\mathbb{P}}(1/\sqrt{T}).
\end{equation}
Now, we have the ingredients to derive the asymptotic normality of our averaged OIRF estimator
$$
\hat{\bar{\theta}}_{r}^{q}(i)= 
\widehat{\Phi}_i^{als} \frac{q+h}{q}  Int_{T,r, q+h} \left(C\left(\widehat V\right)\right),
$$
of the averaged OIRF $  \bar{\theta}_{r}^{q}(i) =  \Phi_i Int_{r,q}(C(\Sigma)).
$
First, note that by \eqref{utile_q} and the $\sqrt{T}-$convergence of $\mbox{vec}(\widehat{\Phi}_i^{als})$
\begin{multline*}
\sqrt{T}\mbox{vec} \left(   \hat{\bar{\theta}}_{r}^{q}(i) -  \bar{\theta}_{r}^{q}(i) \right) = \mbox{vec}\bigg[\sqrt{T} \left(\widehat{\Phi}_i^{als}  - \widehat{\Phi}_i  \right)  Int_{r,q}(C(\Sigma))\\ \left. + \widehat{\Phi}_i \sqrt{T} \left\{ \frac{q+h}{q}  Int_{T,r, q+h} \left(C\left(\widehat V\right)\right) - Int_{r,q}(C(\Sigma)) \right\} \right]+ o_{\mathbb{P}}(1).
\end{multline*}
By \eqref{iid_rep_c}, the $\sqrt{T}-$asymptotic normality of
\begin{multline*}
( I_d\otimes \Phi_i ) \left\{ \!\frac{q+h}{q}  \mbox{vec}\left (Int_{T,r,q+h} (C (\widehat V)) \right) \! -\mbox{vec}  \left(Int_{r,q}\left(C\left(\Sigma\right)\right) \right)\! \right\} \\
 = ( I_d\otimes \Phi_i ) \left\{ \mbox{vec}\left (\bar{\widehat{H}}(r)  \right) -\mbox{vec}  \left(Int_{r,q}\left(C\left(\Sigma\right)\right) \right) \right\}
\end{multline*}
follows from the CLT applied to
\begin{multline*}
\sqrt{T} ( I_d\otimes \Phi_i ) \left\{ Int_{T,r,q-h}  (\Delta\mbox{vec}(\widehat U)) - Int_{r,q-h}  (\Delta\mbox{vec}(\Sigma))\right\} \\
= \sqrt{T} ( I_d\otimes \Phi_i ) \Delta \left\{\mbox{vec}( Int_{T,r,q-h}  (\widehat U)) - \mbox{vec}(Int_{r,q-h} (\Sigma))\right\}\\
= \sqrt{T} ( I_d\otimes \Phi_i ) \Delta \left\{\mbox{vec}( Int_{T,r,q-h}  (U)) - \mbox{vec}(Int_{r,q-h} (\Sigma))\right\}\{1+o_{\mathbb P }(1)\}
\\
= \sqrt{T} ( I_d\otimes \Phi_i ) \Delta \left\{\mbox{vec}( Int_{T,r,q}  (U)) - \mbox{vec}(Int_{r,q} (\Sigma))\right\}\{1+o_{\mathbb P }(1)\}.
\end{multline*}
The result follows from the  $\sqrt{T}-$asymptotic normality of $\mbox{vec}(\widehat{\Phi}_i^{als})$ and the zero-mean condition for the product of any three components of the error vector, see Assumption \textbf{A1}. \end{proof}

\vspace{0.3cm}

Let us note that, taking $A(\cdot)$ equal to the identity matrix $I_d$ in \eqref{equiv_crucial} we can deduce   that  new approximated OIRF estimator  could be equivalently defined, with $\widetilde q = q$, as equal to
$$
 \sqrt{\frac{q+2h}{q}} \; \widehat{\Phi}_i^{als}C(Int_{T,r,q+2h} (\widehat V)) ,
$$
where here $\widehat V$ is defined in \eqref{nnapprx2}.
The difference between the two definitions is asymptotically negligible. More precisely,
\begin{multline*}
\hat{\widetilde{\theta}}_{r}^{q,als}(i)=
\widehat{\Phi}_i^{als} C\left(Int_{T,r,q} \left(\widehat U\right)\right)
\\ =
 \sqrt{\frac{q+2h}{q}} \; \widehat{\Phi}_i^{als}C(Int_{T,r,q+2h} (\widehat V)) +
 O_{\mathbb{P}}(h^2+1/Th)\\
= \sqrt{\frac{q+2h}{q}} \; \widehat{\Phi}_i^{als}C(Int_{T,r,q+2h} (\widehat V)) +
 o_{\mathbb{P}}(1/\sqrt{T}),
\end{multline*}
provided $Th^4+1/Th^2 \rightarrow 0$.

\vspace{0.3cm}

\begin{proof}[Proof of Proposition \ref{asy_index}]
In the sequel, when we use the $o_{\mathbb{P}}(\cdot)$ and $O_{\mathbb{P}}(\cdot)$
symbols for a vector or a matrix, it should be understood as used for their norms.
Recall that
\begin{equation*}
i_{r,q} =  \left\| \bar H(r)^{-1}\widetilde H(r) \right\|_2^2 ,
\end{equation*}
where
$$
\bar{H}(r)= Int_{r,q} (C(\Sigma))
\quad
\text{and}
\quad
\widetilde H(r) = C(Int_{r,q} (\Sigma)).
$$
The estimator we propose is
$$
\widehat i_{r,q} =
\left\| \bar{\widehat {H}}(r)^{-1}\; \widehat{\widetilde H}(r) \right\|_2^2
$$
where
$$
  \bar{\widehat{H}}(r)=   \frac{1}{[ q T]+1}  \sum_{k=-[(q + h)T/2]}^{[(q+ h)T/2]}\widehat{H}_{[rT]-k}= \frac{q+h}{q} Int_{T,r,q+h}(C(\widehat V))
$$
and
$$
\widehat{\widetilde H}(r) = C (\widehat S_T(r))
$$
with $\widehat S_T(r)$ some estimator of $q^{-1}\int_{r-q/2}^{r+q/2}\Sigma(v)dv$.

By \eqref{id78}
\begin{multline*}
 \mbox{vec}\left( \bar{\widehat{H}}(r) \right) =Int_{r,q} \left(\mbox{vec} \left(C\left(\Sigma\right)\right) \right) \\
 +\frac{q-h}{q} \left\{ Int_{T,r,q-h}  (\Delta\mbox{vec}( U)) - Int_{r,q-h}  (\Delta\mbox{vec}(\Sigma))\right\} + O_{\mathbb{P}}(h^2+1/Th)\\
 =Int_{r,q} \left(\mbox{vec} \left(C\left(\Sigma\right)\right) \right)
 +\frac{q-h}{q} \left\{  \Delta\mbox{vec}( Int_{T,r,q-h} (U)-Int_{r,q-h}(\Sigma))\right\} + o_{\mathbb{P}}(1/\sqrt{T})
 \\ = : \mbox{vec}\left( \bar{H}(r) \right)  + \frac{q-h}{q} G_{T,r,q-h} + o_{\mathbb{P}}(1/\sqrt{T}),
\end{multline*}
with $\Delta$ defined in \eqref{diff_Ch}. If we consider
$$
\widehat S_T(r) =  Int_{T,r,q}  (\widehat  U) =  Int_{T,r,q}  (  U) + o_{\mathbb P}(1/\sqrt{T})
$$
and use the identity
\begin{multline*}
\mbox{vec} (C (\widehat S_T(r))) - \mbox{vec} (C (Int_{r,q}(\Sigma))) = \Delta \left[ \mbox{vec} (\widehat S_T(r) )-  \mbox{vec} \left(Int_{r,q}(\Sigma))\right) \right] \{1+o_{\mathbb{P}} (1)\}\\
= \Delta \left[ \mbox{vec} (Int_{T,r,q}  (  U)-Int_{r,q}(\Sigma))\right]  \{1+o_{\mathbb{P}} (1)\},
\end{multline*}
we deduce
$$
 \mbox{vec}\left( \widehat{\widetilde{H}}(r) \right) = \mbox{vec}\left( \widetilde{H}(r) \right)  + G_{T,r,q} + o_{\mathbb{P}}(1/\sqrt{T}).
$$
Note that
$$
\frac{q-h}{q} G_{T,r,q-h} - G_{T,r,q} =  O_{\mathbb{P}}(h/\sqrt{T}).
$$

We deduce from above
\begin{multline*}
\bar{\widehat {H}}(r)^{-1}\; \widehat{\widetilde H}(r) = \left[ I_d+ \bar{H}(r)^{-1}\left\{\bar{\widehat {H}}(r)-\bar{H}(r)\right\} \right]^{-1} \!\bar{H}(r)^{-1} \!\left[  \widetilde H(r) + \left\{\widehat{\widetilde H}(r)-\widetilde H(r) \right\}\right]\\
= \left[ I_d-  \bar{H}(r)^{-1}\left\{\bar{\widehat {H}}(r)-\bar{H}(r)\right\} + O_{\mathbb{P}}(1/T) \right]\bar{H}(r)^{-1} \left[  \widetilde H(r) + \left\{\widehat{\widetilde H}(r)-\widetilde H(r) \right\}\right]\\
= \bar{H}(r)^{-1}   \widetilde H(r) + \bar{H}(r)^{-1}\left\{\widehat{\widetilde H}(r)-\widetilde H(r)\right\}-\bar{H}(r)^{-1}\left\{\bar{\widehat {H}}(r)-\bar{H}(r)\right\}\bar{H}(r)^{-1}   \widetilde H(r)\\ + O_{\mathbb{P}}(1/T)\\
= \bar{H}(r)^{-1}   \widetilde H(r) + \bar{H}(r)^{-1}\left\{\mbox{ivec} (G_{T,r,q})\right\}\!-\frac{q-h}{q}\bar{H}(r)^{-1}\!\left\{\mbox{ivec} (G_{T,r,q-h})\right\}\bar{H}(r)^{-1}   \widetilde H(r)\\ + O_{\mathbb{P}}(h^2+1/Th+1/T),
\end{multline*}
where $\mbox{ivec}(\cdot)$ denotes the inverse of the $\mbox{vec}(\cdot)$ operator: for any matrix $A$, $\mbox{ivec} (\mbox{vec} (A))=A$. In particular, we deduce that in the case where $\Sigma(\cdot)$ is constant on the interval $[r-q/2, r+q/2]$, and thus  $i_{r,q} = 1$, we have
$$
\left\| \bar{\widehat {H}}(r)^{-1}\; \widehat{\widetilde H}(r) - I_d \right\|_2 =  o_{\mathbb P}(1/\sqrt{T}).
$$
As a consequence,
\begin{multline*}
\left| \; \left\| \bar{\widehat {H}}(r)^{-1}\; \widehat{\widetilde H}(r) \right\|_2 - 1\; \right| = \left|
\left\| \bar{\widehat {H}}(r)^{-1}\; \widehat{\widetilde H}(r) \right\|_2 - \left\|I_d \right\|_2\right| \\ \leq \left\| \bar{\widehat {H}}(r)^{-1}\; \widehat{\widetilde H}(r) - I_d \right\|_2 =  o_{\mathbb P}(1/\sqrt{T}),
\end{multline*}
and thus
$$
 \widehat i_{r,q} - 1 = o_{\mathbb P}(1/\sqrt{T}).
$$

In the case where $\Sigma(\cdot)$ is not constant on the interval $[r-q/2, r+q/2]$, and thus $i_{r,q} > 1$, let us note that
$i_{r,q}$ is also the largest eigenvalue of the symmetric matrix
 \begin{equation}\label {matrix6}
 {\widetilde H} (r) ^\prime \bar { H}(r) ^{-1\prime }  \bar { H}(r) ^{-1} {\widetilde H} (r) .
 \end{equation}
 By the decomposition of $\bar{\widehat {H}}(r)^{-1}\; \widehat{\widetilde H}(r)$ we have
 \begin{multline*}
\widehat{\widetilde H} (r) ^\prime \bar {\widehat H}(r) ^{-1\prime }  \bar {\widehat H}(r) ^{-1} \widehat{\widetilde H} (r)\\ =
 \left\{\bar{H}(r)^{-1}   \widetilde H(r) + M_{T,r,q} + o_{\mathbb P}(1/\sqrt{T})\right\}^\prime  \times \left\{\bar{H}(r)^{-1}   \widetilde H(r) + M_{T,r,q} + o_{\mathbb P}(1/\sqrt{T})\right\}
 \\
 = \widetilde H (r) ^\prime \bar {H}(r) ^{-1\prime }  \bar {H}(r) ^{-1} {\widetilde H} (r) + \mathcal H_{T,r,q} + o_{\mathbb P}(1/\sqrt{T}),
 \end{multline*}
 where
 $$
 \mathcal H_{T,r,q} = M_{T,r,q}^\prime \bar{H}(r)^{-1}   \widetilde H(r)  + \widetilde H (r) ^\prime \bar {H}(r) ^{-1\prime }M_{T,r,q}  ,
 $$
$$
M_{T,r,q} = \bar{H}(r)^{-1}\left\{\mbox{ivec} (G_{T,r,q})\right\}-\bar{H}(r)^{-1}\left\{\mbox{ivec} (G_{T,r,q})\right\}\bar{H}(r)^{-1}   \widetilde H(r)
$$
and, recall,  $G_{T,r,q}= \Delta\mbox{vec}( Int_{T,r,q-h} (U)-Int_{r,q-h}(\Sigma)).$
By the delta-method and the differential of the first eigenvalue of a symmetric matrix, see Theorem 7, section 8, Magnus and Neudecker (1988),
$$
\sqrt{T}\left( \widehat i_{r,q} - i_{r,q} \right) = \upsilon_1 ^\prime \sqrt{T}\mathcal H_{T,r,q}\upsilon_1 + o_{\mathbb P}(1) = (\upsilon_1 ^\prime \otimes \upsilon_1 ^\prime) \mbox{vec}(\sqrt{T}\mathcal H_{T,r,q})+ o_{\mathbb P}(1),
$$
with $\upsilon_1$ a normalized eigenvector associated to the largest eigenvalue $i_{r,q}$ of the matrix \eqref{matrix6}. Finally,
CLT guarantees that $\mbox{vec}(\sqrt{T}\mathcal H_{T,r,q})$ convergences in distribution to a Gaussian limit. The result follows.  \end{proof}

\newpage

\section*{References}
    \begin{description}
\item[]{\sc Aue, A., H\"{o}rmann S., Horv\`{a}th L.
    and Reimherr, M.} (2009) Break detection in the covariance structure of multivariate time series models.
    \textit{Annals of Statistics} 37, 4046-4087.
\item[] {\sc Alter, A., and Beyer, A.} (2014) The dynamics of spillover effects during the European sovereign debt turmoil. \textit{Journal of Banking and Finance} 42, 134-153.
\item[] {\sc Beetsma, R., and Giuliodori, M.} (2012) The changing macroeconomic response to stock market volatility shocks. \textit{Journal of Macroeconomics} 34, 281-293
\item[] {\sc Benkwitz, A., L\"{u}tkepohl, H., and Neumann, M.H.} (2000) Problems related to confidence intervals for impulse response of autoregressive processes. \textit{Econometric Review} 19, 69-103.
\item[] {\sc Bernanke, B.S., and Mihov, I.} (1998a) Measuring monetary policy. \textit{The Quarterly Journal of Economics} 113, 869-902.
\item[] {\sc Bernanke, B.S., and Mihov, I.} (1998b) The liquidity effect and long-run neutrality. \textit{Carnegie-Rochester Conference Series on Public Policy} 49, 149-194.
\item[] {\sc Blanchard, O., and Simon, J.} (2001) The long and large decline in U.S. output volatility. \textit{Brookings Papers on Economic Activity} 1, 135-164.
\item[] {\sc Cavaliere, G., Rahbek, A., and Taylor, A.M.R.} (2010) Testing for co-integration in vector autoregressions with non-stationary volatility. \textit{Journal of Econometrics} 158, 7-24.
\item[]{\sc Cavaliere, G., and Taylor, A.M.R.} (2007) Time-transformed unit-root tests for models with non-stationary volatility. \textit{Journal of Time Series Analysis} 29, 300-330.
\item[]{\sc Cavaliere, G., and Taylor, A.M.R.} (2008) Bootstrap unit root tests for time Series with nonstationary volatility. \textit{Econometric Theory} 24, 43-71.
\item[] {\sc Chang, X.-W., and Stehl{\'e}, D.} (2010)  Rigorous Perturbation Bounds of Some Matrix Factorizations.
\textit{SIAM J. Matrix Analysis Applications} 31, 2841--2859.

%

\item[]{\sc Dahlhaus, R.} (1997) Fitting time series models to nonstationary processes. \textit{Annals of
    Statistics} 25, 1-37.
\item[]{\sc Dees, S., and Saint-Guilhem, A.} (2011) The role of the United States in the global economy and its evolution over time. \textit{Empirical Economy} 41, 573-591.
\item[]{\sc Diebold, F., and Yilmaz, K.} (2014) On the network topology of variance decompositions: measuring the connectedness of financial firms. \textit{Journal of Econometrics} 182, 119-134

\item[]{\sc Flury, B.N.} (1985) Analysis of linear combinations with extreme
ratios of variance. \textit{Journal of the American Statistical Association}
80, 915-922.
\item[]{Giraitis, L., Kapetanios, G., and Yates, T.} (2018) Inference on multivariate heteroscedastic time varying random coefficient models. \textit{Journal of Time Series Analysis} 39, 129-149.

\item[]{\sc Kew, H., and Harris, D.} (2009) Heteroskedasticity-robust testing for a fractional unit root. \textit{Econometric Theory} 25, 1734-1753.

\item[]{\sc L\"{u}tkepohl, H.} (2005) \!\textit{New Introduction to Multiple Time Series Analysis}. Springer, Berlin.
\item[]{\sc L\"{u}tkepohl, H., Staszewska-Bystrova, A., and Winker, P.} (2015) Confidence bands for impulse response:
Bonferroni vs. Wald. \textit{Oxford Bulletin of Economics and Statistics} 77, 800-821.
\item[]{\sc Magnus, J.R.,  and Neudecker, H.} (1988)  \emph{Matrix differential calculus with applications in
statistics and econometrics.} John Wiley \& Sons.
\item[]{\sc Nazlioglu, S., Soytas, U., and Gupta, R.} (2015) Oil prices and financial stress: A volatility spillover analysis. \textit{Energy Policy} 82, 278-288.
\item[] {\sc Patilea, V., and Ra\"{i}ssi, H.} (2013) Corrected portmanteau tests for VAR models with time-varying variance. \textit{Journal of Multivariate Analysis} 116, 190-207.
\item[] {\sc Patilea, V., and Ra\"{i}ssi, H.} (2012) Adaptive estimation of vector autoregressive models with time-varying variance: application to testing linear causality in mean. \textit{Journal of Statistical Planning and Inference} 142, 2891-2912.
\item[] {\sc Patilea, V., and Ra\"{i}ssi, H.} (2010) Adaptive estimation of vector autoregressive models with time-varying variance: application to testing linear causality in mean. Working paper, arXiv:1007.1193v2.
\item[] {\sc Patilea, V., and Ra\"{i}ssi, H.} (2014) Testing second-order dynamics for autoregressive processes in presence of time-varying variance. \textit{Journal of the American Statistical Association} 109, 1099-1111.
\item[] {\sc Primiceri, G.E.} (2005) Time varying structural vector autoregressions and monetary policy. \textit{The Review of Economic Studies} 72, 821-852.
\item[]{\sc  Ra\"{i}ssi, H.} (2015) Autoregressive order identification for VAR models with non-constant variance. \textit{Communications in Statistics: Theory and Methods} 44, 2059-2078.
\item[]{\sc Sensier, M., and van Dijk, D.} (2004) Testing for volatility changes in U.S. macroeconomic time series. \textit{Review of Economics and Statistics} 86, 833-839.
\item[]{\sc Sims, C.A.} (1999) Drift and breaks in monetary policy. \textit{Unpublished paper}.
\item[]{\sc Stock, J.H., and Watson, M.W.} (2002)  Has the business cycle has changed and why? \textit{NBER Macroeconomics Annual} 17, 159-230.
\item[]{\sc Stock, J.H., and Watson, M.W.} (2005) Understanding changes in international business cycle dynamics. \textit{Journal of the European Economic Association} 3, 968-1006.
\item[]{\sc Strongin, S.} (1995) The identification of monetary
policy disturbances explaining the liquidity puzzle. \textit{Journal of Monetary Economics} 35, 463-497.
\item{\sc Strohsal, T., Proa\~{n}o, C.R., and Wolters, J.} (2019) Assessing the Cross-Country Interaction of Financial Cycles: Evidence from a Multivariate Spectral Analysis of the US and the UK. \textit{Empirical Economics} 57, 385-398.
\item{\sc van der Vaart, A., and Wellner, J.A.} (2011). A local maximal inequality under
uniform entropy. \textit{Electronic Journal of Statistics}  5,  192-203.
\item[]{\sc Xu, K.L., and Phillips, P.C.B.} (2008)
    Adaptive estimation of autoregressive models with time-varying
    variances. \textit{Journal of Econometrics} 142, 265-280.
\end{description}

\newpage

\section*{Tables and Figures}

\begin{table}[hh]\!\!\!\!\!\!\!\!\!\!
\begin{center}
\caption{The kernel estimators of $\int g$ and $(\int g^2)^{0.5}$ for the  monthly global price of brent crude, in U.S. Dollars per barrel. The pre and post crisis periods are from January 1990 to July 2008, and from January 2020 to July 2020.}
\begin{tabular}{|c|c|c|}
  \hline
  periods & pre-crisis & post-crisis \\
  \hline
  $(\int g^2)^{0.5}$ & 3.91 & 5.48  \\
  \hline
  $\int g$ & 3.75 & 3.61  \\
  \hline
\end{tabular}
\label{tab00}
\end{center}
\end{table}

\begin{table}[hh]\!\!\!\!\!\!\!\!\!\!
\begin{center}
\caption{The $\widehat{i}_{r,q}$'s for the oil-inflation data.}
\begin{tabular}{|c|c|c|}
  \hline
  periods & pre-crisis & post-crisis \\
  \hline
  $\widehat{i}_{r,q}$ & 1.26 & 1.63 \\
  \hline
\end{tabular}
\label{tab5}
\end{center}
\end{table}

\begin{figure}[h]
\begin{center}
\includegraphics[scale=0.46]{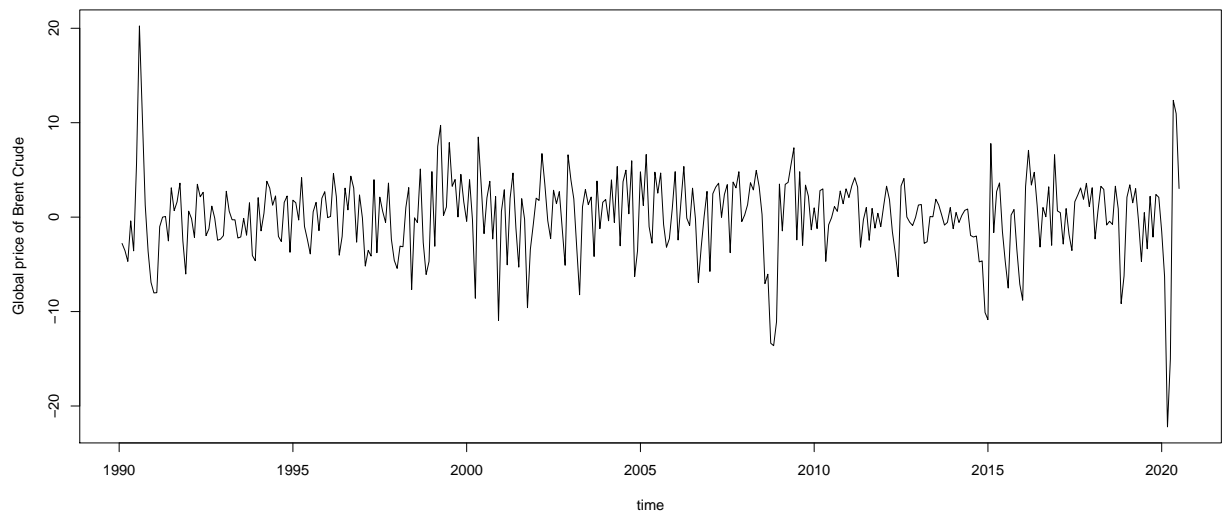}
\end{center}
\caption{The log differences of the monthly global price of brent crude multiplied by 100, in USD per barrel, from January 1990 to July 2020.}
\label{fig1}
\end{figure}

\begin{figure}[h]
\begin{center}
\includegraphics[scale=0.46]{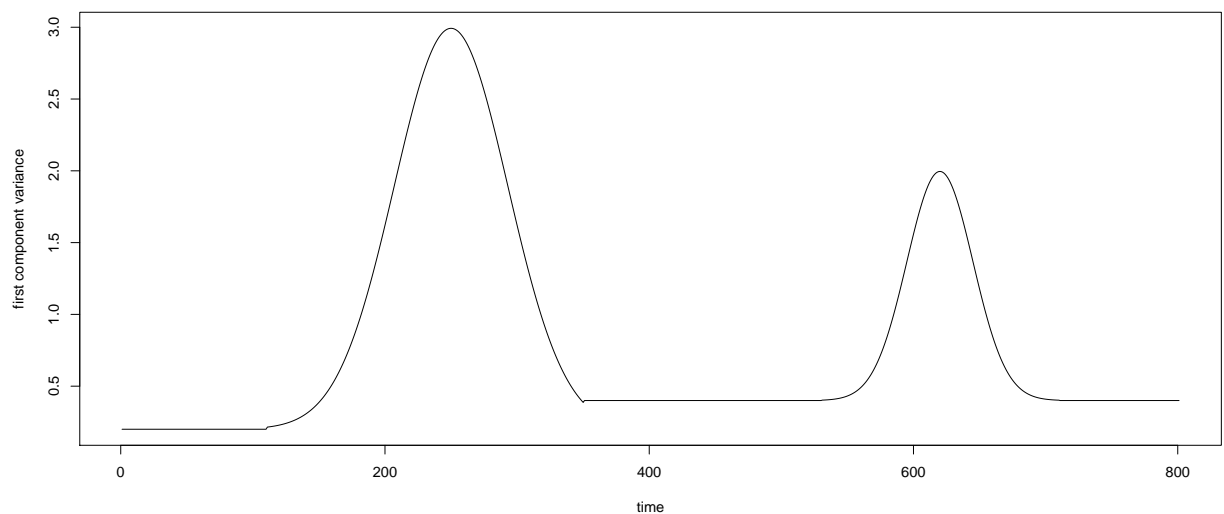}
\end{center}
\caption{The variance structure $\sigma^2_{11}(r)$ of the first innovations component of the simulated process (\ref{simulDGP}).}
\label{first-var-comp}
\end{figure}

\begin{figure}[h]
\begin{center}
\includegraphics[scale=0.46]{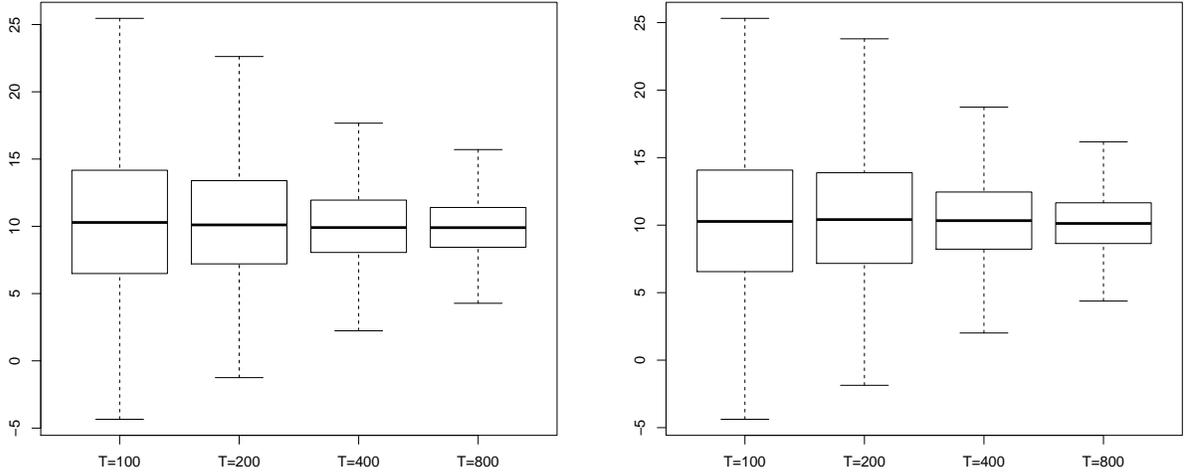}
\end{center}
\caption{The relative differences between the approximated and averaged OIRFs: $100*\left(\frac{\hat{\widetilde{\theta}}_{0.5}^{0.5,11}(1)}{\hat{\bar{\theta}}_{0.5}^{0.5,11}(1)}-1\right)$, (see equations (\ref{ALSapproxhist}) and (\ref{accumulatedgls_H})). The results corresponding to a bandwidth with $T^{-1/3}$ (resp. $T^{-2/7}$) decreasing rate is displayed on the left (resp. on the right).}
\label{fig2}
\end{figure}

\begin{figure}[h]
\begin{center}
\includegraphics[scale=0.46]{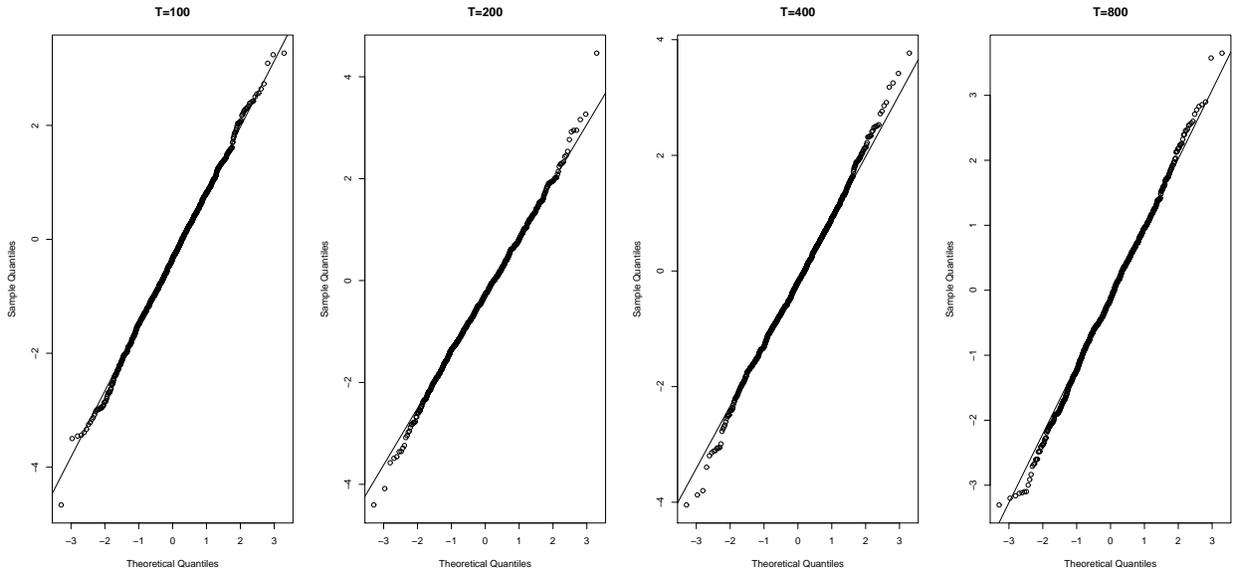}
\end{center}
\caption{The normal Q-Q plot of the approximated OIRFs of order one, that is $\sqrt{T}(\hat{\widetilde{\theta}}_{0.5}^{0.5,11}(1)-\widetilde{\theta}_{0.5}^{0.5,11}(1))$', over the $N=1000$ iterations.}
\label{fig3}
\end{figure}

\begin{figure}[h]
\begin{center}
\includegraphics[scale=0.46]{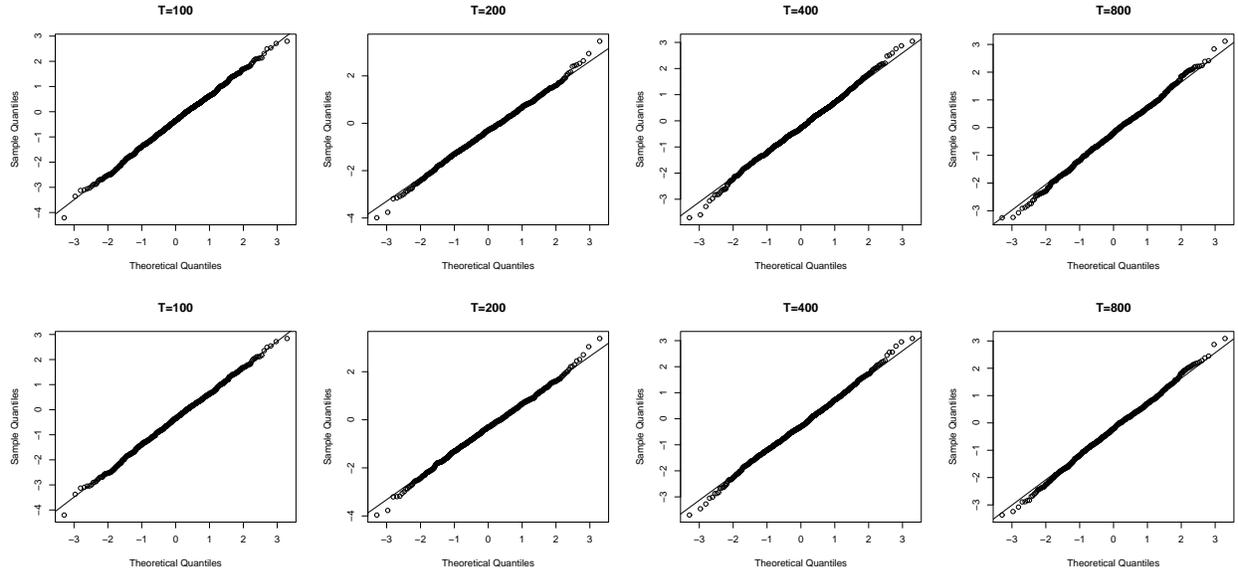}
\end{center}
\caption{The normal Q-Q plot of the averaged OIRFs of order one $\sqrt{T}(\hat{\bar{\theta}}_{0.5}^{0.5,11}(1)-\bar{\theta}_{0.5}^{0.5,11}(1))$'s. The results corresponding to a bandwidth with a $T^{-1/3}$ (resp. $T^{-2/7}$) decreasing rate are displayed on the top (resp. on the bottom) panels.}
\label{fig4}
\end{figure}

\begin{figure}[h]
\begin{center}
\includegraphics[scale=0.46]{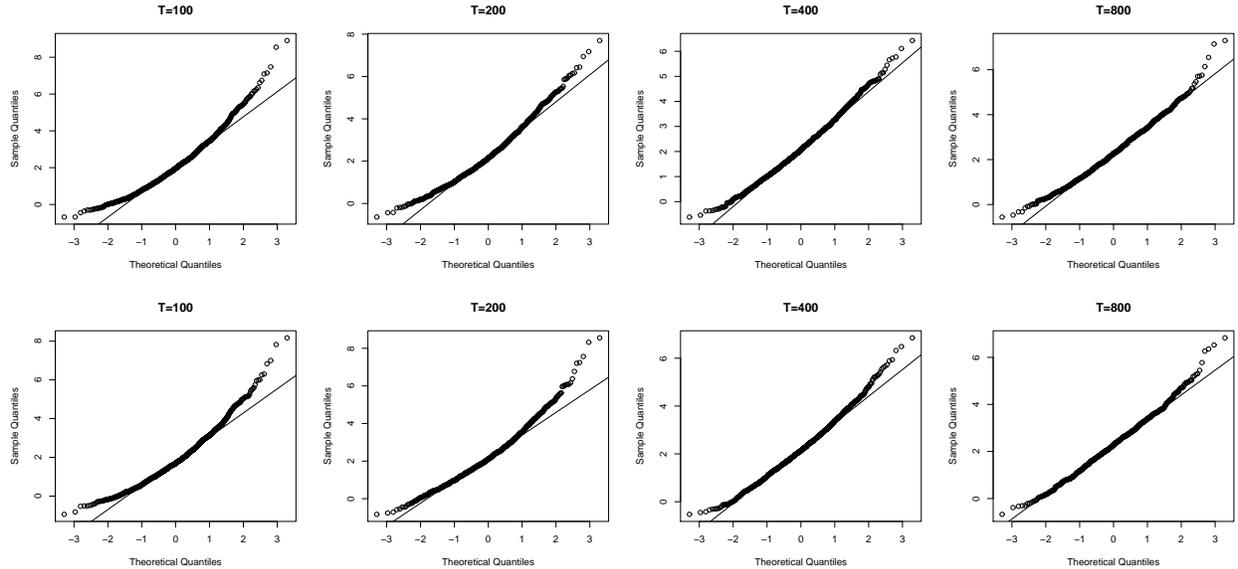}
\end{center}
\caption{The normal Q-Q plot of the
$\sqrt{T}\left(\widehat i_{r,q} - i_{r,q} \right)$'s. The results corresponding to a bandwidth with a $T^{-1/3}$ (resp. $T^{-2/7}$) decreasing rate are displayed on the top (resp. on the bottom) panels.}
\label{fig5}
\end{figure}

\begin{figure}[h]
\begin{center}
\includegraphics[scale=0.64]{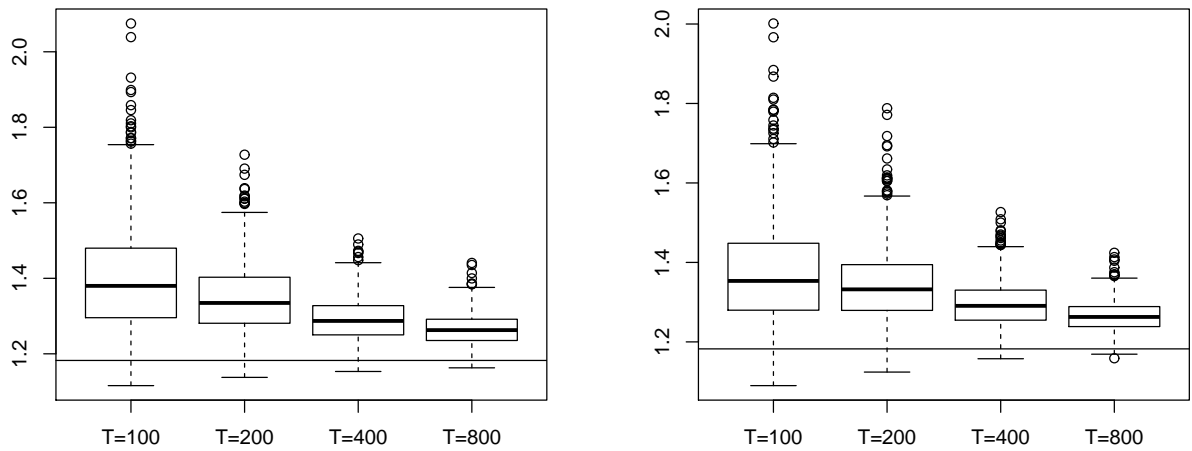}
\end{center}
\caption{The box-plots of the $\widehat i_{r,q}$'s for different sample sizes. The horizontal line corresponds to the true value. The results corresponding to a bandwidth with a $T^{-1/3}$ (resp. $T^{-2/7}$) decreasing rate are displayed on the left (resp. on the right) panels.}
\label{fig6}
\end{figure}

\begin{figure}[h]
\begin{center}
\includegraphics[scale=0.46]{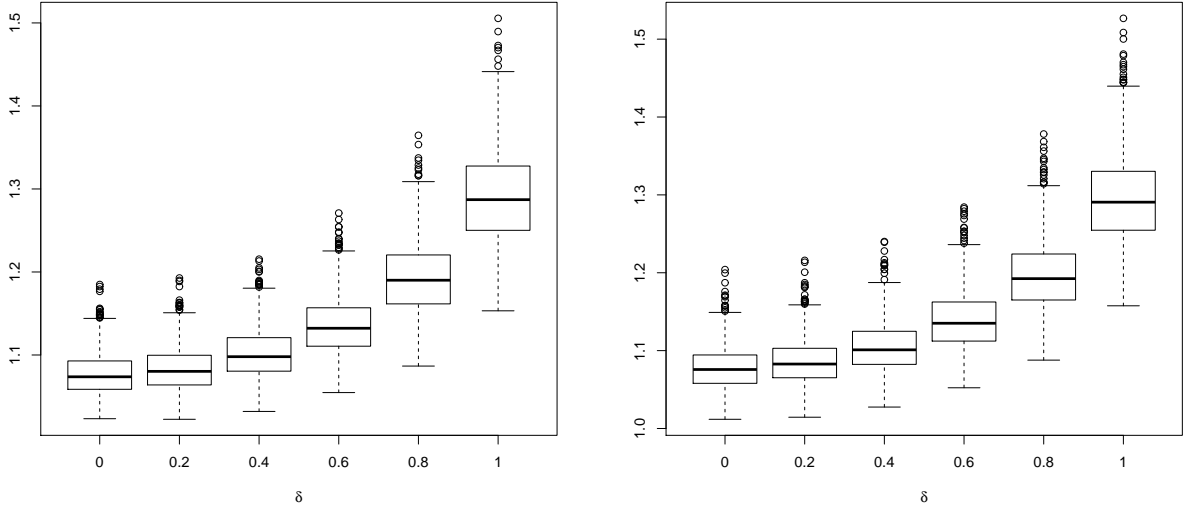}
\end{center}
\caption{The box-plots of the $\widehat i_{r,q}$'s for different values for the heteroscedasticity parameter $\delta$. As $\delta$ is far from zero, the heteroscedasticity is more marked. The results corresponding to a bandwidth with a $T^{-1/3}$ (resp. $T^{-2/7}$) decreasing rate are displayed on the left (resp. on the right) panel.}
\label{fig7}
\end{figure}

\begin{figure}[h]
\begin{center}
\includegraphics[scale=0.46]{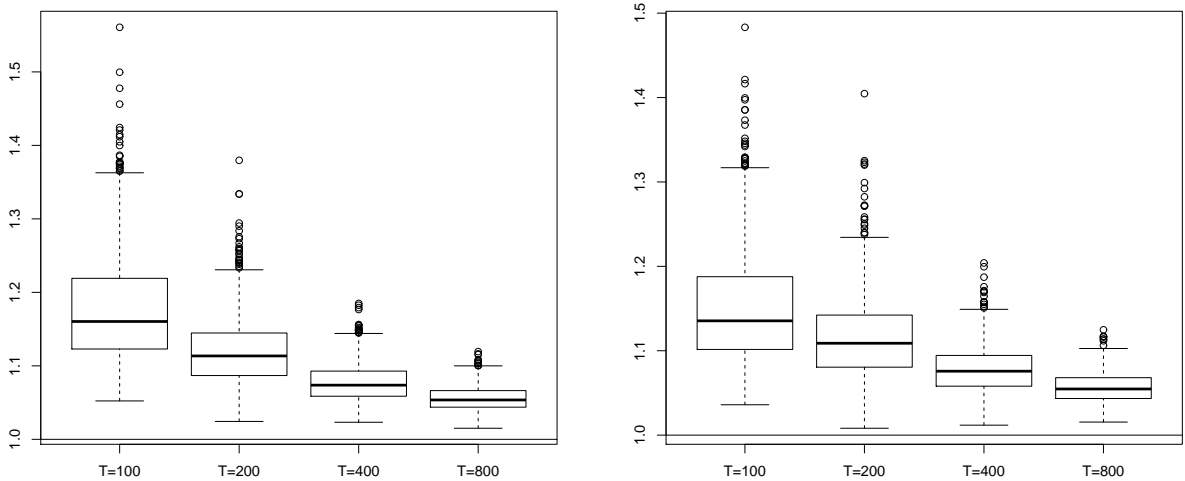}
\end{center}
\caption{The box-plots of the $\widehat i_{r,q}$'s for different sample sizes in the homoscedastic case (the true value is equal to one). The results corresponding to a bandwidth with a $T^{-1/3}$ (resp. $T^{-2/7}$) decreasing rate are displayed on the left (resp. on the right) panel.}
\label{fig8}
\end{figure}

\begin{figure}[h]
\begin{center}
\includegraphics[scale=0.45]{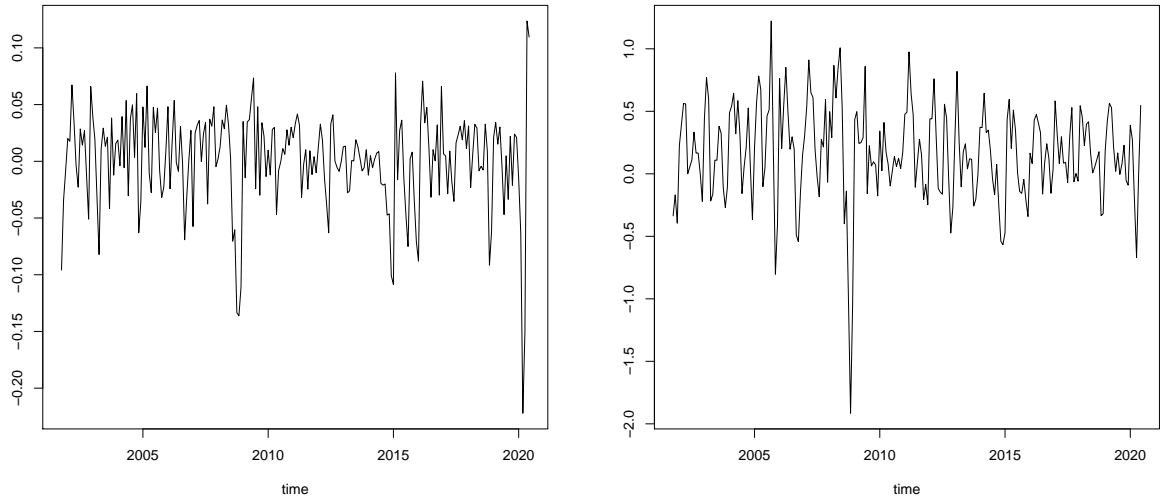}
  \end{center}
\caption{The $\log$ first differences of the brent crude in USD per barrel multiplied by 100, on the left. The growth rate previous period for the consumer price index for the United States on the right. The series are monthly, taken from October, 2001 to June, 2020.}
\label{data}
\end{figure}

\begin{figure}[h]
\begin{center}
\includegraphics[scale=0.45]{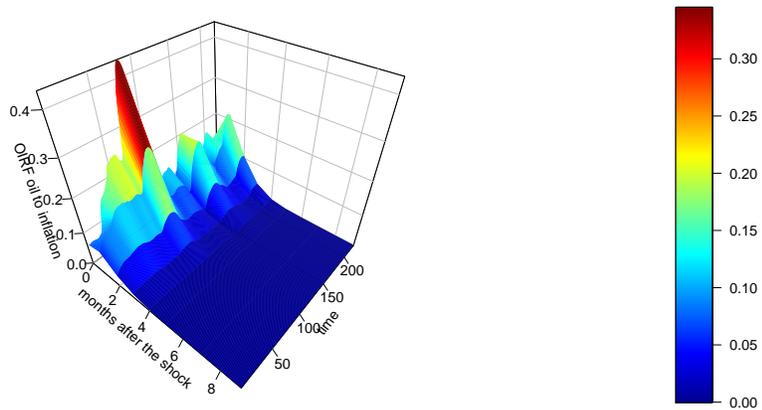}
  \end{center}
\caption{The pointwise orthogonal response of the inflation growth rate response to an oil price shock.}
\label{fig9}
\end{figure}

\begin{figure}[h]
\begin{center}
\includegraphics[scale=0.30]{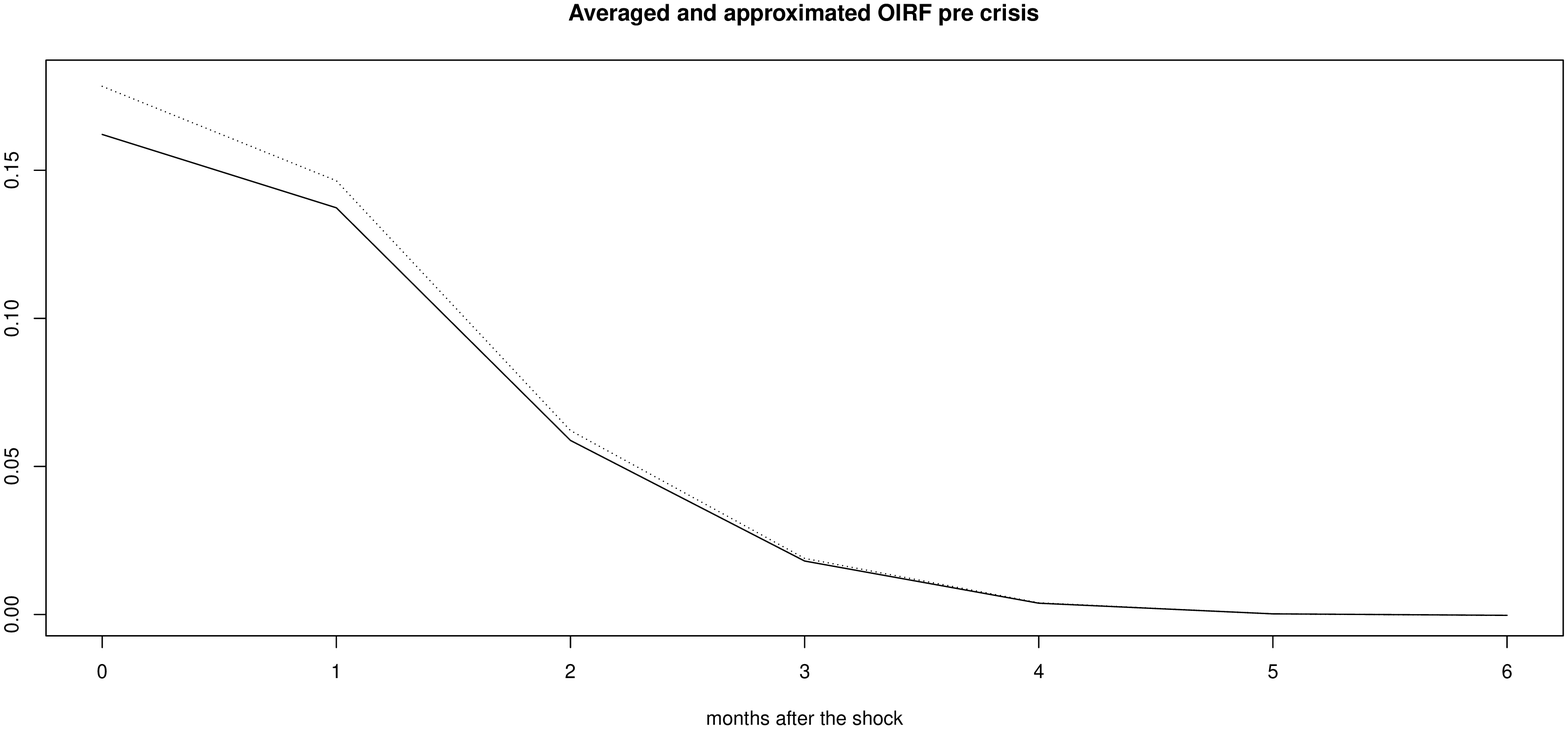}
\includegraphics[scale=0.30]{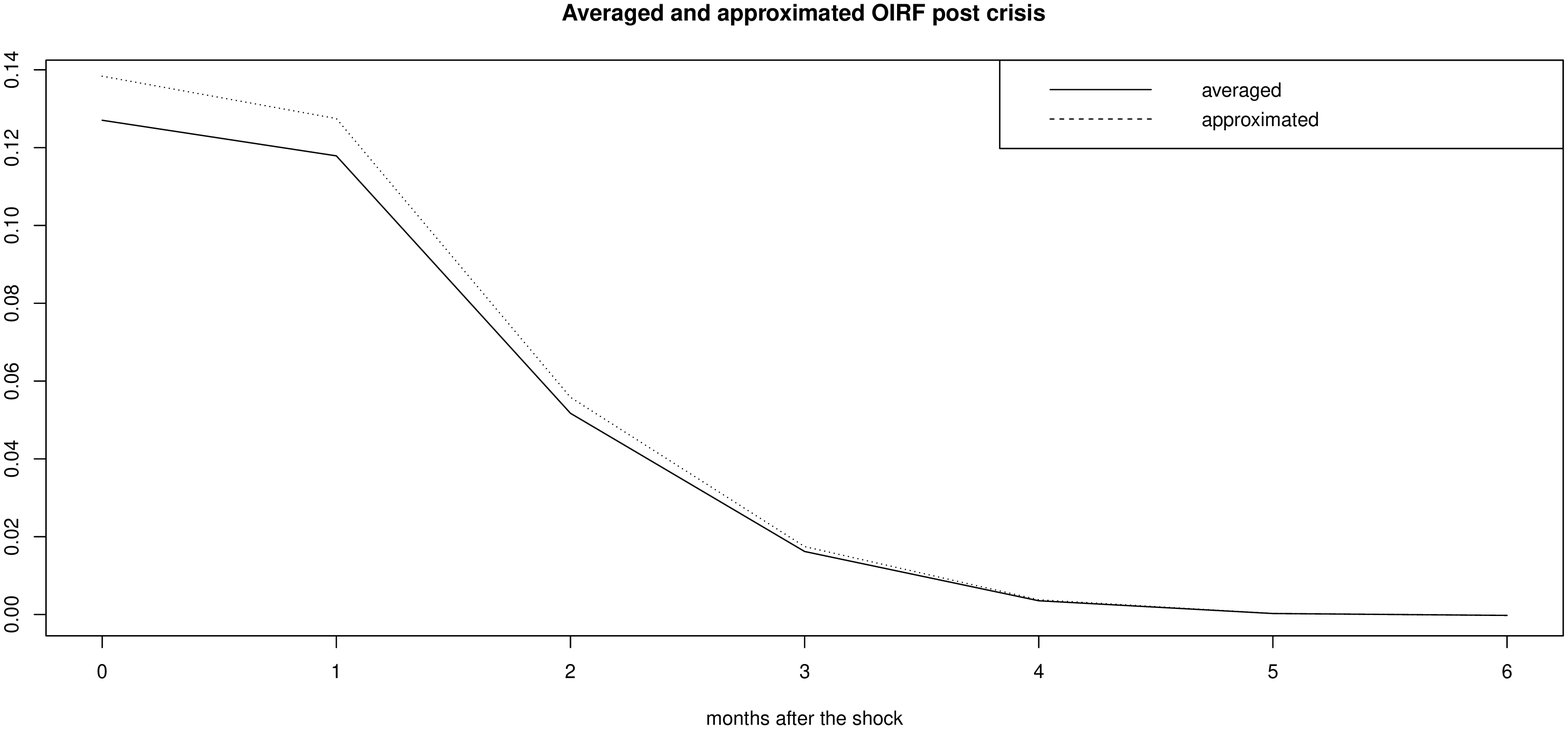}
  \end{center}
\caption{The pre and post crisis averaged and approximated OIRF for the oil prices-inflation variables.}
\label{fig10}
\end{figure}

\end{document}